\documentclass{llncs}

\usepackage{algo}  
\usepackage{epsfig}  
\usepackage{url}  
\usepackage{subfigure}  
\usepackage{latexsym}  
\usepackage{boxedminipage}  
\begin{document}  
  
\newcounter{algo}  
\setcounter{algo}{0}  
\newenvironment{algo}[1]{%
\refstepcounter{algo}  
\begin{center}\sl Algorithm \thealgo: #1  
\begin{boxedminipage}{0.95\textwidth}  
}{\end{boxedminipage}\end{center}}

\newcommand{\mC}{\mathcal{C}}  
\newcommand{\mF}{\mathcal{F}}  

\newenvironment{preuve}{\noindent {\it Proof.}}{$\Box$\vskip1ex}  
\newenvironment{preuveA1}{\noindent {\it Proof of conservation of   
Property $A_1$.}}{$\Box$\vskip1ex}  
\newenvironment{preuveA2}{\noindent {\it Proof of conservation of   
Property $A_2$.}}{$\Box$\vskip1ex}  
\newenvironment{preuveB}{\noindent {\it Proof of the invariants.}}  
{$\Box$\vskip1ex}  
\newenvironment{OJO}{\noindent {\bf OJO:}}{$\Box$\vskip1ex}  
\newcommand{\currfingname}{\mathit{curr\mathunderscore fing\mathunderscore name}}  
\newcommand{\lastname}{\mathit{last\mathunderscore name}}    
\newcommand{\lastchar}{\mathit{last\mathunderscore char}}    
\newcommand{\tmpname}{\mathit{tmp\mathunderscore name}}
\newcommand{\namesubtable}{\mathit{name\mathunderscore sub\mathunderscore table}}
\newcommand{\finglink}{\mathit{fing\mathunderscore link}}   
\newcommand{\fingpar}{\mathit{fing\mathunderscore par}}
\newcommand{\fingchar}{\mathit{fing\mathunderscore char}}
\newcommand{\tmpsubname}{\mathit{tmp\mathunderscore sub\mathunderscore name}}
\newcommand{\subnamepair}{\mathit{sub\mathunderscore name\mathunderscore pair}}
\newcommand{\nameuptable}{\mathit{name\mathunderscore up\mathunderscore table}}
\newcommand{\subname}{\mathit{sub\mathunderscore name}}
\newcommand{\upname}{\mathit{up\mathunderscore name}}
\newcommand{\parentfingname}{\mathit{parent\mathunderscore fing\mathunderscore name}}
\newcommand{\fingname}{\mathit{fing\mathunderscore name}}
\newcommand{\oldname}{\mathit{old\mathunderscore name}}
\newcommand{\currname}{\mathit{curr\mathunderscore name}}
\newcommand{\childlist}{\mathit{child\mathunderscore list}}
\newcommand{\maxlocstart}{\mathit{max\mathunderscore loc\mathunderscore start}}
\newcommand{\maxlocend}{\mathit{max\mathunderscore loc\mathunderscore end}}
\newcommand{\maxloclist}{\mathit{max\mathunderscore loc\mathunderscore list}}
\newcommand{\charlabel}{\mathit{char\mathunderscore label}}
\newcommand{\occ}{\mathit{occ}}

\title{\vspace*{-2cm}Various improvements to text fingerprinting\thanks{This work is supported by the Russian
     Foundation for Fundamental Research (Grant 05-01-00994) and the
     program of the President of the Russian Federation for supporting
     of young researchers (Grant MD-3635.2005.1)} \thanks{This work is also supported by the french ANR project MAPPI.}}   
  
 \author{Djamal Belazzougui\inst{1} \quad Roman
   Kolpakov\inst{2} \quad Mathieu
   Raffinot\inst{1}} \institute{LIAFA, Univ. Paris Diderot - Paris 7,
   75205 Paris Cedex 13, France,\\ {\tt \{dbelaz,
     raffinot\}@liafa.jussieu.fr} \and Liapunov French-Russian
   Institute, Lomonosov Moscow State University, Moscow, Russia, {\tt
     foroman@mail.ru} }

\maketitle

\noindent
{\small {\bf Abstract:} 
Let $s = s_1 .. s_n$ be a text (or sequence) on a finite alphabet
$\Sigma$ of size $\sigma$. A fingerprint in~$s$ is the set of distinct characters
appearing in one of its substrings. The problem considered here is
to compute the set ${\cal F}$ of all fingerprints of all substrings of~$s$
in order to answer efficiently certain questions on this set.  
A substring $s_i .. s_j$ is a maximal location for a fingerprint $f\in F$ 
(denoted by $\langle i,j \rangle$) if the alphabet of $s_i .. s_j$ is~$f$ 
and $s_{i-1}$, $s_{j+1}$, if defined, are not in~$f$. The set of
maximal locations in~$s$ is ${\cal L}$ (it is easy to see that $|{\cal L}| \leq n \sigma$). 
Two maximal locations $\langle i,j \rangle$ and $\langle k,l \rangle$ such
that $s_i .. s_j = s_k .. s_l$ are named {\em copies}, and the quotient set 
of ${\cal L}$ according to the copy relation is denoted by ${\cal L}_C$.  

We first present new exact efficient algorithms and data structures for the following three
problems: (1) to compute ${\cal F}$; (2) given $f$ as a set of distinct
characters in $\Sigma$, to answer if $f$ represents a fingerprint in
${\cal F}$; (3) given $f$, to find all maximal locations of $f$ in~$s$. 
As well as in papers concerning succinct data structures, in the paper 
all space complexities are counted in bits.
Problem 1 is solved either in $O(n+|{\cal L}_C|\log \sigma)$ worst-case time 
(in this paper all logarithms are intended as base two logarithms) 
using $O((n+|{\cal L}_C|+|{\cal F}|\log \sigma)\log n)$ bits 
of space, or in $O(n+|{\cal L}|\log \sigma)$ randomized expected time 
using $O((n+|{\cal F}|\log \sigma)\log n)$ bits of space. 
Problem 2 is solved either in $O(|f|)$ expected time 
if only $O(|f|\log n)$ bits of working space for queries is allowed, or
in worst-case $O(|f|/\epsilon)$ time if a working
space of $O(\sigma^{\epsilon}\log n)$ bits is allowed (with $\epsilon$ a constant satisfying $0<\epsilon<1$). 
These algorithms use a data structure that occupies $|{\cal F}|(2\log\sigma+\log_2 e)(1+o(1))$ bits. 
Problem 3 is solved with the same time complexity as Problem 2, but with the addition of an $\occ$ term
to each of the complexities, where $\occ$ is the number of maximal locations
corresponding to the given fingerprint. Our solution of this last problem 
requires a data structure that occupies $O((n+|{\cal L}_C|)\log n)$ bits of memory. 

In the second part of our paper we present a novel Monte Carlo approximate construction
approach. Problem 1 is thus solved in $O(n+|{\cal L}|)$ expected time using $O(|{\cal
F}|\log n)$ bits of space but the algorithm is incorrect
with an extremely small probability that can be bounded in advance.


\section{Introduction}  
\label{sec:intro}

We consider a finite ordered alphabet $\Sigma$ with $\sigma =|
\Sigma|$ and $s=s_1..s_n$ a sequence of $n$ letters, $s_i \in
\Sigma$. The set of all sequences over $\Sigma$ is denoted
$\Sigma^*$. The rank of each letter $\alpha$ in $\Sigma$ is given by
$f_{\Sigma}(\alpha)$ that ranges between $0$ and $\sigma-1$.  A
sequence $v \in \Sigma^* $ is a factor or substring of $s$ if
$s=uvw$. The fingerprint $C(s)$ of a sequence $s$ is the set of
distinct letters in $s$. By extension, $C_s(i,j)$ is the set of
distinct letters in $s_i..s_j.$


\begin{definition}
Let ${\cal C}$ be a set of letters of $\Sigma$. A maximal location of
${\cal C}$ in $s=s_1 .. s_n$ is an interval $[i,j]$, $1\leq i \leq
j \leq n$,  such that\\
$\begin{array}{lll}
(1)~~ C_s(i,j) = {\cal C}; &  ~~ (2) ~~ \mbox{if } i>1, s_{i-1} \not \in
  C_s(i,j); &  ~~ (3) ~~ \mbox{if } j<n, s_{j+1} \not \in C_s(i,j)\\
\end{array}$\\
This maximal location is denoted $\langle i,j \rangle.$
\vspace{-2mm} 
\end{definition}

\noindent
We denote by ${\cal F}$ the set of distinct fingerprints and by ${\cal
  L}$ the set of maximal locations of all fingerprints of ${\cal F}$.

\begin{definition}
Two maximal locations $ \langle i,j \rangle$ and $ \langle k,l
\rangle$ of $s=s_1..s_n$ are {\em copies} if $s_i..s_j$ $=$
$s_k..s_l.$
\end{definition}

\noindent
The ``copy'' relation is obviously an equivalence relation over ${\cal
  L}$, and we denote ${\cal L}_{C}$ the set of equivalence classes.
In this paper, given a sequence $s$, we are interested in three main
problems:

\begin{enumerate}
\item Compute the set ${\cal F}$ of all fingerprints in $s$;
\item Given a fingerprint $f$, find whether $f$ is a
  fingerprint in ${\cal F}$;
\item Given a fingerprint $f$, find all the maximal locations
of $f$ in $s$.
\end{enumerate}


Efficient answers to these questions have many applications in
information retrieval, computational biology and natural language
processing \cite{AmirALS03}. The input alphabet $\Sigma$ is considered
to be the alphabet of the input sequence, thus $\sigma\leq
n$. Notice that $|{\cal L}| \leq n \sigma$. The best current
algorithms solve Problem 1 in $\Theta(\min \{ n + |{\cal L}| \log
\sigma, n^2 \})$ time and space. The bound $\Theta(n + |{\cal L}|
\log\sigma)$ is that of \cite{KR08b}.  The $\Theta(n^2)$ bound is
obtained using the algorithm of Didier et al.~\cite{DSST04}. Problem 2 is
solved in $O(|f|\log (\sigma/|f|))$ time and $O(|{\cal F}|)$ space 
($O(|{\cal F}|\log n)$ bits) and Problem 3 in 
$O(|f|\log (\sigma/|f|) + occ)$ time (where $occ$ is the number of maximal locations
that match the given fingerprint) and $O(|{\cal F}|+ |{\cal L}|)$ space 
($O((|{\cal F}|+|{\cal L}|)\log n)$ bits) in \cite{CYHW07,CYHW11}.

We first present new exact efficient algorithms and data structures for the three
problems we considered above. 

Problem 1 is solved either in $O(n+|{\cal L}_C|\log \sigma)$ worst-case time using 
$O((n+|{\cal L_C}|+|{\cal F}|\log \sigma)\log n)$ bits of space, or in 
$O(n+|{\cal L}|\log\sigma)$ randomized expected time using $O((n+|{\cal F}|\log \sigma)\log n)$ bits of space.

Problem 2 is solved either in $O(|f|)$ expected time and space if only $O(|f|\log n)$ bits of working space for queries is allowed, 
or in $O(|f|/\epsilon)$ worst-case time  if a working space of $O(\sigma^{\epsilon}\log n)$ bits is allowed. 
This problem uses a data structure which occupies $|{\cal F}|(2\log\sigma+\log_2 e)(1+o(1))$ bits. 
Previous and new exact results are summarized in table~\ref{table:problem2}.

Problem 3 is solved in the same time as Problem 2, with the addition of an $\occ$ term 
to each of the complexities, where $\occ$ is the number of maximal locations
corresponding to the fingerprint searched. Previous and new exact results are summarized in tables 1-3.

\begin{table}
\noindent
\begin{tabular}{|l|l|l|}
  \hline
  Solution & Build space (bits)& Build time \\
 \hline 
  prev. \cite{KR08b} (worst-case) & $O((n+|{\cal L}|)\log n)$  & $O(n+|{\cal L}|\log\sigma)$ \\
 \hline
  theorems~\ref{theo:build_part_tree},\ref{theo:build_names} (worst-case)& 
  $O((n+|{\cal L}_C|+|{\cal F}|\log\sigma)\log n)$& $O(n+|{\cal L_C}|\log\sigma)$\\
  \hline
  theorem~\ref{theo:rand_build} (randomized expected)& 
$O((n+|{\cal F}|\log\sigma)\log n)$ & $O(n+|{\cal L}|\log\sigma))$ \\
  \hline
  theorem~\ref{theo:MC} (Monte-Carlo)& $O((n+|{\cal F}|)\log n)$ & $O(n+|{\cal L}|)$\\

\hline 
\end{tabular}
\caption{Previous and new solutions to Problem 1 (Determination of ${\cal F}$).}
\label{table:problem1}
\end{table}

\begin{table}
\noindent
\centering
\begin{tabular}{|l|l|l|}
  \hline
  Solution & Data structure space (bits)& Query time \\
 \hline 
  prev. & $O(|{\cal F}|\log n)$ & $O(|f|\log (\sigma/|f|))$\\
  \hline
  theorem~\ref{theo:succ_exist_DS}& $O(|{\cal F}|\log\sigma)$ & $O(|f|)$\\
\hline 

\end{tabular}
\caption{Previous and new solution for Problem 2 (existential fingerprint queries).}
\label{table:problem2}
\end{table}

\begin{table}

\noindent
\centering
\begin{tabular}{|l|l|l|}
 \hline
  Solution & Data structure space (bits)& Query time \\

\hline 
  prev. & $O(|{\cal L}| \log n)$ & $O(|f|\log (\sigma/|f|))+\occ)$\\
\hline
  theorem~\ref{theo:fast_rep_DS}& $O(|{\cal L}| \log n)$ & $O(|f|+\occ)$\\
\hline
  theorem~\ref{theo:fast_rep_DS}& $O((n+|{\cal L}_C|) \log n)$ & $O(|f|+\occ)$\\

\hline 
\end{tabular}
\caption{Previous and new solutions to Problem 3 (maximal location report queries)}
\label{table:problem3}
\end{table}

\vspace{-0.5cm}
\noindent
In this article we also propose a novel Monte Carlo approximate query
approach. The result of the query may not be exact, but an error occurs at a
probability that one can fix {\em a priori} as small as
required. This approach has the advantage of speeding up the
identification of all fingerprints by a $\log \sigma$ factor. Problem
1 is thus solved in $O(n+|{\cal L}|)$ expected time using $O(|{\cal F}|\log n)$ 
bits of space using a Monte Carlo approach, but the algorithm yields
incorrect results with an extremely low probability. Table
\ref{table:problem3} summarizes the complexities of the construction
space and time including the Monte-Carlo method.

Our algorithms are based on several tools of four main natures:
hash functions, succinct data structures, trees, and naming techniques first introduced in \cite{KMR72}, 
adapted to the fingerprint problem in \cite{AmirALS03} and then
successively improved in \cite{DSST04} and in \cite{KR08b}. These
tools are presented in Section \ref{tools}. In Section \ref{redundant}
we present our $O(n+|{\cal L}_C|\log \sigma)$ worst-case time construction algorithm. 
Section \ref{space}
presents a more space efficient representation of ${\cal F}$ in space
$O(|{\cal F}|\log \sigma)$ bits instead of $O(|{\cal F}|\log n)$ bits. 
This data structure allows us to solve Problem 2 and 3
in the complexities bounds announced above.
Then Section \ref{less} contains the $O(n+|{\cal L}|\log\sigma)$ expected
time algorithm using $O((n+|{\cal F}|\log\sigma)\log n)$-bit space 
for solving Problem 1. Finally, in
Section \ref{random} we present the Monte Carlo algorithm that allows
us to efficiently solve Problem 1 in time $O(|{\cal L}|)$ and space
$O(|{\cal F}|\log n)$ thus saving a $\log\sigma$ factor in both 
space and time complexity of the algorithm in section \ref{less}.

We assume below without loss of generality that the input sequence
does not contain two consecutive repeating characters. Such a sequence
is named {\em simple}. The segments of repeating characters, say
$\alpha$, of any input sequence can be reduced to a unique occurrence
of $\alpha$. The two sequences have the same set ${\cal F}$ and the
same sets ${\cal L}$ and ${\cal L}_C$, up to small changes in the
bounds (these changes can be simply retrieved in $\Theta(1)$ time
per maximal location and produced by trivial algorithm in $\Theta(n)$ time). 
This technical trick greatly simplifies the algorithms we present by
removing many straightforward technical cases.

All the algorithms presented in this paper assume the unit-cost word RAM model with word length $w=\Omega(\log n)$ 
and with usual arithmetic and logic operations taking constant time (additions, multiplication, bitwise operations etc.).

\section{Tools}
\label{tools}

This section is devoted to the four main tools we use in our
algorithms, namely polynomial hash functions, the suffix tree, the
participation tree and the naming technique.

\subsection{Hash functions}
\label{subsec:poly_hash}
Our constructions are based on the use of polynomial hash functions
modulo $P$, where $P$ is a suitably chosen prime. Given a collection
$M$ of $m$ sets over a universe $\sigma$, our goal is to find a
polynomial hash function so that each set is mapped to a distinct
value. The polynomials are evaluated modulo an arbitrary prime $P$ chosen such
that $m^2\sigma \leq P \leq 2m^2\sigma$ (we will show later how 
to efficiently find such a prime). More precisely, we will use a family of hash
functions $H_P=\{h_X|X\in [1,P-1]\}$, where each hash function $h_X\in
H_P$ in the family is parametrized with an integer $X\in [1,P-1]$. The
functions of the family are defined in the following way : for any set
$S$ of $t$ distinct integers $S=\{e_1,e_2,\ldots,e_t\}$ such that
$S\subseteq[0,\sigma-1]$ we have: $$h_X(S)=\sum_{i=1}^{t}X^{e_i}\bmod
P$$ In order to compute a fixed hash function $h_X$ on any set $S$ in
$O(|S|)$ time, we can use a precomputed table of size $\sigma$, which
stores all the powers of $X$ up to $X^{\sigma-1}$. Alternatively, we
could use a two-dimensional precomputed table $T$ of size
$c\cdot \lceil\sigma^{1/c}\rceil$ for any integer $c$ ensuring a computation
time of $O(c|S|)$. That is, we store in $T[i,j]$ the number
$X^{i\gamma^j}$ where $\gamma=\lceil\sigma^{1/c}\rceil$. Then in order
to compute $X^{e_i}$, we can use the property that $e_i$ can be
decomposed into a sum of $c$ numbers
: $$e_i=\sum_{j=0}^{c-1}d_{ij}\gamma^j$$ where each $d_{ij}$ can be
computed using the formula: $$d_{ij}=\lfloor
e_i/\gamma^j\rfloor\bmod\gamma$$ Thus for computing $X^{e_i}$, it
suffices to use the
formula: $$X^{e_i}=\prod_{j=0}^{c-1}X^{d_{ij}\gamma^j}=\prod_{j=0}^{c-1}T[d_{ij},j]$$
To summarize, given any set $S=\{e_1,e_2,\ldots,e_t\}$ where
$S\subseteq[0,\sigma-1]$, $h_X(S)$ can be computed in $O(c\cdot t)$
time. First, for each $e_i$, compute $X^{e_i}$ in $O(c)$ time: for
each $e_i$ compute its decomposition $\sum_{0\leq j<c}d_{ij}\gamma^j$
in $O(c)$ time where each $d_{ij}$ is computed by $d_{ij}=\lfloor
e_i/\gamma^j\rfloor\bmod\gamma$, and then compute $X^{e_i}$ also in
$O(c)$ time using the formula $X^{e_i}=\prod_{0\leq j< c}T[d_{ij},j]$. Thus, the computations of all $X^{e_i}$ take $O(c\cdot t)$
time in total. The final step is to sum all of the computed $X^{e_i}$
which takes time $O(t)$.  \\Summarizing, for any set $S$ of $t$
elements the computation of $h_X(S)$ takes $O(c\cdot t)$. The space needed
by the precomputed table $T$ is $O(c\cdot \sigma^{1/c})$.  \\ In the
following we will need the technical lemma below:
\begin{lemma}
\label{lemma:poly_hash_lemma}
Given a collection $M$ of $m$ integer sets where each set is a subset of $[0,\sigma-1]$, a randomly chosen hash function $h_X \in H_P$ 
for $P\geq m^2\sigma$ will injectively map the collection $M$ to the interval $[0,P-1]$ with probability at least $1/2$. 
\end{lemma}
\begin{proof}
The lemma is easy to prove. Take any pair of sets $(x,y)\in M^2$. The two sets $x$ and $y$ are mapped to the same hash value 
by a function $h_X\in H_P$ if and only if $(h_X(x)-h_X(y))=0$. Now $h_X(x)-h_X(y)$ is a polynomial of degree at most $\sigma-1$ 
over the field $GF[P]$ which consequently can have at most $\sigma-1$ roots. 
Therefore for any pair $(x,y)\in M^2$ we have that $(h_X(x)-h_X(y))$ can possibly be zero for at most $\sigma-1$ different values of $X$. 
As we have $m(m-1)/2$ such pairs, the number of values of $X$ for which we have a collision for any of the pairs is at most $t=(\sigma -1)m(m-1)/2$. 
We have $P=\sigma m^2$ and therefore $t\leq P/2$. 
\end{proof}
We now sketch how to efficiently find one prime number in the interval $[m^2\sigma,2m^2\sigma]$. By well known properties of 
the distribution of prime numbers, we know that the density of primes below a given number $N$  is roughly logarithmic in $N$. 
This suggests the following simple algorithm: randomly pick a number $P$ in the interval $[m^2\sigma,2m^2\sigma]$. 
The number $P$ will be prime with probability $\Omega(1/\log (m^2\sigma))=\Omega(1/(\log m +\log\sigma))$. 
Then test whether $P$ is a prime using any efficient deterministic primality testing algorithm that takes time polylogarithmic in $P$. 
If $P$ is not a prime, then repeat the same procedure (pick a random $P$ in the interval and test its primality) until we get a prime $P$. 
Because the probability of $P$ being prime is $\Omega(1/(\log m +\log\sigma))$, the expected number of repeated procedures will be $O(\log m+\log\sigma)$. 
As a primality testing takes time polylogarithmic in $(m^2\log\sigma)$ and we are doing $O(\log m+\log\sigma)$ expected primality tests, 
we deduce that the total time for finding $P$ is $O((\log m+\log\sigma)^c)$ for some constant~$c$. 

\subsection{Succinct Data Structures}

\subsubsection{Succinct Static Function Representation}
We will make use of the following recent result:
\begin{lemma}~\cite{P09}
\label{lemma:succ_fn_rep}
Given a set $S\subseteq U$ where $|U|\leq 2^w$,$|S|\geq \log|U|$ and a function $f$ from $S$ into $[0,2^k-1]$ (with $k\leq w$), we can, in $O(|S|)$ time build a succinct representation of the function $f$ that uses $|S|k(1+o(1))$ bits.
Given any element $x\in S$ the representation returns $f(x)$ in constant time. Given an element $x\in U\backslash S$, the representation returns an arbitrary value in $[0,2^k-1]$ in constant time. 
\end{lemma}
The result stated in the lemma was first described in~\cite{P09}. 
It combines the use of a set of hash functions with matrix solving on
$GF[2^k]$ (two similar methods are also described in~\cite{CC08,DP08}
but have slightly worse performance). 
The lemma says that we can have a representation
of a function $f$ from $S\subseteq U=[0,2^w-1]$ into $[0,2^k-1]$ that can
successfully return the correct value for $f(x)$ when queried for an
element $x\in S$, but returns an arbitrary value for any element $x$
outside $S$. Therefore, the representation is unable to detect whether a
given element $x$ is in $S$ or not. This is why the space usage in the
lemma has no dependence on $U$, but instead only depends on $k$ and on the
cardinality of $S$ (it is easy to see that in order to detect whether
$x\in S$ we need to store $S$ in one way or another and thus need to use a space of at least $\Omega(|S|\log |U|)$ bits).
 
\subsubsection{Succinctly Encoded Tries (Cardinal trees) }
A trie (or cardinal tree) is a tree where each edge has a label from the alphabet $\Sigma$. The maximal degree in a trie 
is thus $\sigma=|\Sigma|$. A standard representation of a trie of $N$ nodes would need $O(N\log N)$ bits (essentially the $\log N$ bits are needed to encode pointers in the trie). In our case we need a succinct representation that uses less than $O(N\log N)$ bits, ideally close to the information theoretic lower bound which is about $N\log\sigma+O(N)$ bits. We will thus use the following result described in~\cite{RRS07}:
\begin{lemma}
\label{lemma:succ_trie_lemma}
Given a trie (cardinal tree) having a total of $N$ nodes over an alphabet of size $\sigma\geq 2$, we can build a representation that uses $N(\log\sigma+\log_2 e+o(1))$ bits of space and supports basic navigation operations in constant time. In particular it supports the following operation in constant time: given a node $p$ having identifier $i_p$ and a character $\alpha$, tell whether $p$ has a child $d$ labeled with character $\alpha$ and return its identifier $i_d$. 
\end{lemma}
The operation cited in the lemma is the only one which will be used in this paper.
\subsection{Trees}
\subsubsection{Suffix Tree}

The suffix tree $\mbox{ST}(s)$ is a compact representation of all
suffixes of a given sequence $s=s_1 \ldots s_n$. It is basically a
trie of all suffixes of $s$ where all the nodes with a single child are
merged with their parents. Each transition of the tree is then coded
as an interval $[i,j]$ corresponding to $s_i..s_j$. Its size is $O(n)$ 
and it can be built in $O(n)$ time even on integer alphabet using the 
construction algorithm of \cite{Far97}. An example of such a suffix tree 
is given in Figure \ref{suffixtree}.

\begin{figure}[tb]
  \centering
\includegraphics{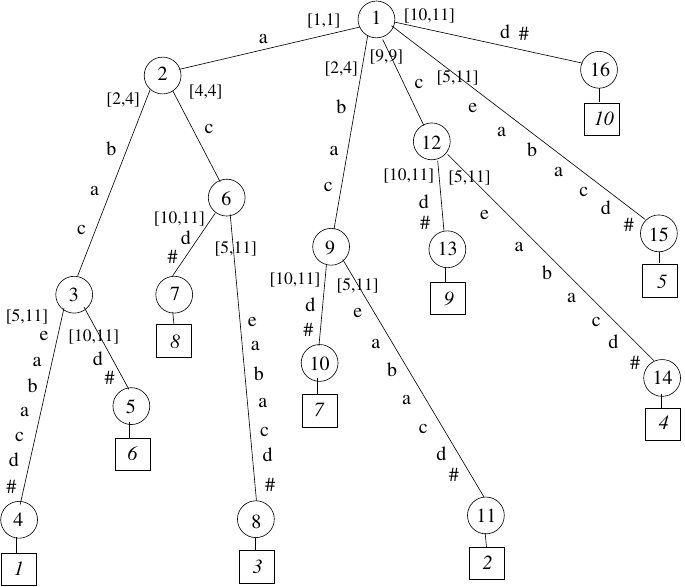}
\caption{Suffix tree of $s = a_1\; b_2\; a_3\; c_4\; e_5\; a_6\; b_7\;
a_8\; c_9\; d_{10} \#_{11}.$ Square boxes contain the initial position
of the suffix. Each edge is labeled by a pair $[k,l]$ pointing to
$s_k..s_l$ that we explicitly write on the edge for clarity.}
 \label{suffixtree}
\vspace*{-0.5cm}
\end{figure}

We assume below that in the suffix tree each transition interval
$[i,j]$ of $\mbox{ST}(s)$ corresponds to the leftmost occurrence of
the factor $s_i\ldots s_j$ in $s$. For instance, in Figure
\ref{suffixtree}, the transition from 1 to 2 is the pointer $[1,1] =
s_1 = a$. This property is ensured by Ukkonen \cite{Ukk92b} algorithm,
but can also be ensured on every suffix tree by a simple additional
$O(n)$ steps.

\subsubsection{Fingerprint Trie}
\label{subsubsec:finger_trie}
We now present the fingerprint trie (this is called backtracking tree
in \cite{CYHW07,CYHW11}).  The fingerprint trie is a tree representation of
the fingerprints. The trie representation exploits the property that
for every $f\in \cal{F}$ such that $|f|\geq 2$ there exists necessarily at
least one other fingerprint $g\in\cal{F}$ and some letter $\alpha$
such that $g\cup \{\alpha\}=f$. In other words, for every $f\in
\cal{F}$ there exists some $g\in\cal{F}$ such that $f$ can be written
as a sequence $\beta_0..\beta_j,\alpha$ (of distinct characters)
and $g\in\cal{F}$ written as a sequence $\beta_0..\beta_j$. 
This property means that the set of
fingerprints can be represented as a trie. More precisely, let
$F_i\subseteq \cal{F}$ be the subset of the fingerprints of $\cal{F}$
where each $f\in F_i$ is of size $i$. At the beginning, we start with
a trie which contains only a root. Then we take the subset $F_1$ 
of all fingerprints in ${\cal F}$ consisting of one character.  
Then for each fingerprint $f\in F_1$ consisting of a character $\alpha$, 
we create a new node corresponding to~$f$ and attach it as a child of 
the root with a link labeled with the character~$\alpha$.  
Then the remainder of the trie can be built level-by-level: for
building level $i\geq 2$, we consider the set ${\cal F}_i$ and for
each $f\in F_i$ do the following:
\begin{enumerate}
\item First consider a fingerprint $g\in F_{i-1}$ (represented by a
  node $q_g$) and a character $\alpha$ such that $g\cup \{\alpha\}=f$
  (by the property above there exists at least one such pair
  $(g,\alpha)$). If there exist several such pairs choose one
  arbitrarily.
\item Then create a new node $q_f$ and attach it as a child of the
  node $q_g$ (which corresponds to $g$) with a link labeled with
  character $\alpha$.
\end{enumerate}

\subsection{Naming Technique}

The naming technique is used to give a
unique name to each fingerprint from~$\cal{F}$. We assume for 
simplicity, but without loss of generality, that
$\sigma$ is a power of two. We consider a stack of $\log\sigma+1$
arrays on top of each other. Each level is numbered from 1. The
lowest, called the fingerprint table, contains $\sigma$ names that
are $[0]$ or $[1]$. Each other array contains half the
number of names that the array it is placed on. The highest array only
contains a single name that will be the name of the whole array. Such
a name is called a fingerprint name. Figure \ref{namingexample} shows
a simple example with $\sigma=8$.

\begin{figure}[htb]
\vspace{-0.3cm}
$$
\begin{array}{|c|c|c|c|c|c|c|c|}
\hline
\multicolumn{8}{|c|}{[7]}\\
\hline
\multicolumn{4}{|c|}{[5]} & \multicolumn{4}{|c|}{[6]}\\
\hline
\multicolumn{2}{|c|}{[2]} & \multicolumn{2}{|c|}{[2]} &
\multicolumn{2}{|c|}{[3]} & \multicolumn{2}{|c|}{[4]}\\
\hline
[1] & [0] & [1] & [0] & [1] & [1] & [0] & [0]\\
\hline
\end{array}
$$
\vspace{-0.5cm}
\caption{Naming example.}
\label{namingexample}
\vspace{-0.3cm}
\end{figure}

The names in the fingerprint table are only $[0]$ or $[1]$ and are
given as input. Each cell $c$ of an upper array represents two cells of the
array it is placed on, and thus a pair of two names. The naming is
done in the following way: for each level going from the lowest to the
highest, if the cell represents a new pair of names, give this pair
a new name and assign it to the cell. If the pair has already been
named, place this name into the cell. In the example in
Figure \ref{namingexample}, the name $[2]$ is associated to $([1],[0])$
the first time this pair is encountered. The second time, this name is
directly retrieved.\\


\noindent
{\bf Naming a List of Fingerprint Changes.}
\label{subchanges}
Assume that a specific set ${\cal S}$ of fingerprints can be
represented as a list $L = (\alpha_1,\alpha_2, \ldots \alpha_p)$ of
distinct characters such that ${\cal S}= \{f_1,f_2,\ldots,f_p \} \mbox{ where
} f_i = \cup_{1 \leq j \leq i} \{ \alpha_j \}.$

The core idea of the
algorithm of \cite{DSST04} is to fill a fingerprint table bottom-up by
building for each level an ordered list of new names that corresponds
to the fingerprint changes induced at the previous level. A
pseudo-code of this naming algorithm is given in Figure
\ref{algo:dekel}.  We explain it below.

\begin{figure}[tb]
\begin{center}
\begin{algorithme}
{\sc Name\_lists}($L=(\alpha_1,\alpha_2, \ldots, \ldots,
\alpha_p)$ initial list of changes)\\ 
\lign $L_1 \leftarrow (\{[0],0\}, \ldots, \{[0],\sigma-1\})$\\ 
\lign add $(\{[1],f_{\Sigma}(\alpha_1)\}, \ldots, \{[1],f_{\Sigma}(\alpha_p)\} )$ to end of $L_1$\\ 
\lign \For{$r = 1..\log \sigma$}\\ 
\lign \> $FT_r \leftarrow$ name table of size $\sigma/2^{r-1}$\\ 
\lign \> $E_{tp} \leftarrow$ first element of $L_r$\\ 
\lign \> \For{$l = 0..\sigma/2^{r-1}-1$}~~~ /* initialization
of table $FT$ */\\ \lign \> \> $\{[a],j\} \leftarrow E_{tp}$\\ 
\lign \> \> $FT_r[j] \leftarrow [a]$\\ 
\lign \> \> $E_{tp} \leftarrow$ next element in $L_r$\\ 
\lign \> \Endfor\\ \lign \> Let $L'_r$ be an empty list\\ 
\lign \> $E_{tp} \leftarrow$ first element of $L_r$\\ 
\lign \> \While{$E_{tp}$ exists}\\ 
\lign \> \> $\{[a],j\} \leftarrow E_{tp}$\\ 
\lign \> \> $FT_r[j] \leftarrow [a]$\\ 
\lign \> \> add $\{(FT_r[2 \lfloor j/2 \rfloor], FT_r[2 \lfloor j/2 \rfloor+1]),\lfloor j/2\rfloor \}$ to end of $L'_r$\\ 
\lign \> \> $E_{tp} \leftarrow$ next element in $L_r$\\ 
\lign \> \Endwhile\\ 
\lign \> sort the pair of names in $L'_r$ in lexicographical order\\ 
\lign \> give new names in each unique pair in $L'_r$\\ 
\lign \> build $L_{r+1}$ by copying $L'_r$ but replacing each pair by its new name\\ 
\lign \Endfor
\end{algorithme}
\end{center}
\vspace{-0.5cm}
\caption{Naming a list $L= (\alpha_1,\alpha_2, \ldots \alpha_p)$ of fingerprint changes.}
\vspace*{-0.5cm}
\label{algo:dekel}
\end{figure}

We number the levels from 1, the lowest, to $\log \sigma +1$. The
original list $L$ is first transformed into a list $L_1$ of changes on
level 1 by replacing each character $\alpha_i$ by the pair
$\{[1],f_{\Sigma}(\alpha_i)\}$. To initialize the process we add a
list of $\sigma$ pairs $\{[0],i\}, \, i = 0.. \sigma-1$ at the beginning
of $L_1$.

This initial list is then used to compute all names of the cells in
the second level. A table $FT$ of $\sigma$ names temporary records
the pair of names to be coded. A list $L'_1$ of pairs of names is
built as follows. The first $\sigma$ elements of $L_1$ are
read to initialize $FT$. The list $L'_1$ is initialized
with $\sigma/2$ pairs built by reading $FT$. Then,
the remainder of the list $L_1$ is read and for each new element
$\{[a],j\}$ {\em (1)} the table $FT$ is changed in position $j$ by $FT[j]
\leftarrow [a]$ and {\em (2)} the pair $\{ (FT[2 \lfloor j/2
\rfloor], FT[2 \lfloor j/2 \rfloor+1]), \lfloor j/2\rfloor \}$ is added to the end of
$L'_1$. This means that in cell $\lfloor j/2\rfloor $ of the
second level a name has to be given to the name pair ($FT[2 \lfloor
j/2 \rfloor], FT[2 \lfloor j/2 \rfloor+1]).$

At this point $L'_1$ records the list of changes to be made in the
cells at level 2 and the pairs of names that must receive a
name. The pairs in this list are then sorted in lexicographical order
(through a radix sort) and a new name is assigned to each distinct
pair of names $(n_1,n_2)$. A new list $L_2$ is built from $L'_1$
(keeping the initial order of $L'_1$ and thus of $L_1$) by replacing
each pair with its new name. For instance, if $\{([1],[0]),1\}$ was in
the list $L'_1$ and if the pair $([1],[0])$ received the new name
$[2]$, then $L_2$ now contains $\{[2],1\}$. 

The list $L_2$ is the input at level 2 and the same process is
repeated to obtain the names in the third level, and so on. The last
list $L_{\log \sigma +1}$ contains the names of all the fingerprints
of ${\cal S}$.\\

\noindent
{\em Complexity.} The sum $\sigma + \sigma/2 + \sigma/4 +\ldots$
(lines 1 and 6-10 of pseudo-code in Fig. \ref{algo:dekel}) for all cell
initializations is bounded by $2\sigma$. The remaining construction of
$L_1$ (line 2) requires $\Theta(|L|)$ time.  Then a linear sort of
$\Theta(|L|)$ elements is performed for every level. As there are
$\log \sigma +1$ levels, naming the list takes $\Theta(\sigma + |L|\log
\sigma)$ time.

\section{Faster Fingerprint Computation}
\label{redundant}

Let $q \in {\cal
L}_{C}$ and $\langle i,j \rangle$ be a maximal location of $q,$ then we denote
$\mbox{st}_s(q)$ as the string $s_i..s_j.$ Table \ref{copytable} shows an
example of a copy relation. Note that the number $|{\cal L}_C|$ can be significantly less than
 $|{\cal L}|$. As an example, we can consider the word $w_k$ over the
 alphabet $\Sigma_k=\{a_1,a_2,\ldots,a_k\}$ which is defined in the
 following inductive way: $w_1=a_1$ and $w_k=w_{k-1} (a_1 a_2\ldots
 a_k)^k $ for $k>1$. For this word we have
 $|w_k|=\frac{1}{6}k(k+1)(2k+1)$, $|{\cal L}|=
 \frac{1}{12}k(3k^3+2k^2-9k+16)=\Theta (|w_k|^{4/3})$, and $|{\cal
   L}_C|=\frac{1}{6}k(k^2+5)=\Theta (|w_k|)$.  Thus, in this case
 $|{\cal L}_C|=o(|{\cal L}|)$ as $k\to\infty$.

\begin{table}[b]
\begin{center}
$\begin{array}{|c|l|l|}
\hline
\mbox{Class } q  & \mbox{ Maximal locations } &  \mbox{st}_s(q)\\
\hline
  I            & \;     \emptyset           &   \;  \varepsilon \\ 
  1            & \;  a_{1}\,|\,a_{3}\,|\,a_{6}\,|\, a_{8} &   \; a\\
  2            & \;    b_{2} \,|\, b_{7}                        &  \;   b \;  \\
 3            & \;     c_{4}   \,|\,   c_{9}                             &  \;  c \;  \\
 4            & \;      d_{10}                                  & \;   d \;  \\
 5            & \;        e_{5}                                & \;   e \;  \\
 6  &\; a_{3}c_{4} \,|\, a_{8}c_{9}  &\; ac  \\
 7 & \;            c_{9}d_{10}  &\; cd  \\
8 & \;            c_{4}e_{5}  &\; ce  \\\hline
\end{array}$\quad
$\begin{array}{|c|l|l|}
\hline
\mbox{Class } q  & \mbox{ Maximal locations } &  \mbox{st}_s(q)\\
\hline       
9 & \;            e_{5}a_{6}  &\; ea  \\
10    & \;     a_{1}b_{2}a_{3}  \,|\,   a_{6}b_{7}a_{8}   &  \; aba  \;  \\
11    & \;     a_{1}b_{2}a_{3}c_4  \,|\,   a_{6}b_{7}a_{8}c_9   &  \; abac  \;  \\
12    & \;     a_{8}c_{9}d_{10} &\; acd \\
13 &\; a_{3}c_{4}e_{5}a_{6} &\; acea  \\
14  &\;  e_5a_6b_7a_8  &\; eaba\\
15    & \;      a_{6}b_{7}a_{8}c_9d_{10}   &  \; abacd  \;  \\
16  &\;  a_{1}b_{2}a_3c_4e_5a_6b_7a_8c_9 &\; abaceabac  \\
17  &\; a_{1}b_{2}a_{3}c_{4}e_{5}a_{6}b_{7}a_{8}c_{9}d_{10}  &\; abaceabcd  \\
\hline
\end{array}$
\end{center}
\vspace*{-0.2cm}
\caption{Copy relation example for $s = a_1\; b_2\; a_3\; c_4\; e_5\; a_6\; b_7\; a_8\; c_9\; d_{10}.$}
\label{copytable}
\end{table}

\subsubsection{Participation Tree}
\label{parttree2}

Let $s=s_1 .. s_n$ be a simple sequence of characters over
$\Sigma$.  In this first phase, for reasons that will become clear
below, we add to the sequence a last character $s_{n+1}=\#$ that does
not appear in the sequence. Thus $s=s_1 .. s_{n}\#_{n+1}$.
Let $i$ and $j$ be positions in $s$, $1 \leq i \leq j \leq n+1$. We define
$\mbox{fo}_s(i,j)$ as the string formed by concatenating the first
occurrences of each distinct character touched when reading $s$ from
position $i$ (included) to position $j$ (included). For instance, if
$s = a_1 b_2 a_3 c_4 e_5 a_6 b_7 a_8 c_9 d_{10}\#$, $\mbox{fo}_s(3,9)=
aceb$ and $\mbox{fo}_s(5,10) = eabcd.$

\begin{definition}
  Let $s=s_1 .. s_n s_{n+1}$ with $s_{n+1}=\#$ and $1\leq i \leq n$ be a
 position in $s$. Let $j>i$ be the minimum position such that $s_j=s_i$
 if it exists, $j=n+2$ otherwise. We define $\mbox{lfo}_s(i) =
 \mbox{fo}_s(i,j-1)$.
\end{definition}

For instance, if $s = a_1b_2c_3a_4d_5a_6b_7a_8c_9b_{10}e_{11} \#_{12}$,
$\mbox{lfo}_s(1) = abc$ and $\mbox{lfo}_s(5) =$ $dabce\#.$

The participation tree resembles a tree of all $\mbox{lfo}_s(i)$ in
which we removed terminal characters (the need of this removal will
appear clearly below). It contains the same path labels. The participation
tree allows some redundancy in the path labels, i.e. the same path label
might correspond to several paths from the root. Thus, our tree is not
always ``deterministic'' in the sense that a node can have several
transitions by the same character. We define it and build it from the
suffix tree by cutting and shrinking edges.

Let $s=s_1..s_n s_{n+1}$ where $s_{n+1}=\#.$ The participation tree
$PT(s)$ is built from the suffix tree $ST(s)$ in the following
way. Imagine the suffix tree in an ``expanded'' version, that is, each
edge $[i,j]$ is explicitly written by the corresponding factor
$s_i..s_j$ (see Figure \ref{suffixtree}).  Let us consider the
sequence of characters on some path from the root and let $\alpha$ be
the first character on this path. Let $o$ be the second occurrence of
$\alpha$ on this path if it exists. We perform the following steps:

\begin{enumerate}
\item We first reduce all characters on this path after $o$
(included) to the empty string $\varepsilon$;
\item Then, on the section from
the root to the character before $o$ we only keep the first occurrence
of each appearing character, i.e. the others are reduced to
$\varepsilon$;
\item We then replace the terminal character of each path from the
root to a leaf by $\varepsilon$;
\item We replace all multi-character edges by an equivalent series of
a single character and a node. An example of such a resulting tree is
shown In Figure \ref{parttree} (left);
\item As a last step, all $\varepsilon$ edges $(p,\varepsilon,q)$ are
removed by merging $p$ and $q$. The resulting tree is the
participation tree.  An example of this last tree is shown in Figure
\ref{parttree} (right).
\end{enumerate}


For each node $q$ of $ST(s)$ and $PT(s)$ we denote by $\mbox{Suff}(q)$
the set of suffixes of $s$ that appear as leaves of the subtree rooted
in $q$. We consider below that the suffixes associated to a node in
$ST(s)$ remain associated to the node in $PT(s)$, even after the
merging. This is shown in Figure \ref{parttree}: the suffixes in the
square boxes associated to nodes $4$ and $5$ in the left picture are
associated to node $2$ in the participation tree (right picture).


\begin{figure}[tb]
\includegraphics{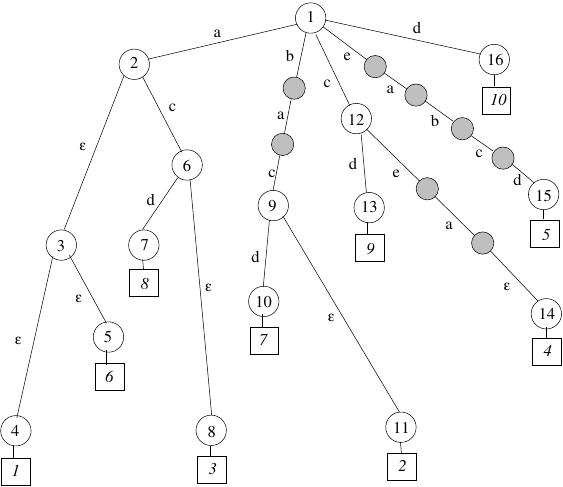} \quad\quad
\includegraphics{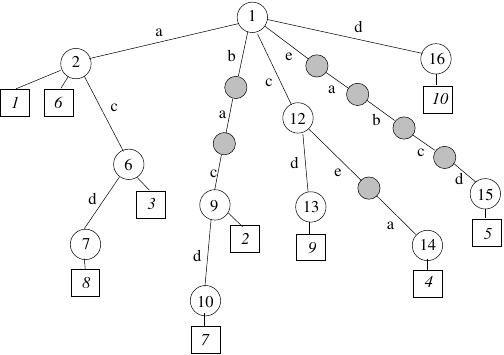}
\caption{From suffix tree to the participation tree (right picture) of
$s = a_1 b_2 a_3 c_4 e_5 a_6 b_7 a_8 c_9 d_{10} \#_{11}.$
New nodes are in gray. The $\varepsilon$ transitions are removed in
the last step. Attached suffixes are shown in square boxes.}
 \label{parttree}
\end{figure}

\begin{lemma}
Let $s= s_1..s_n$. For all $i=1, \ldots ,n$, each proper prefix of $\mbox{lfo}_s(i)$ labels a
path from the root in $PT(s).$
\label{alllfo}
\end{lemma}
\begin{preuve}
When nodes are ignored, the reduction of the path of a suffix $i$ in the suffix tree
corresponds to $\mbox{lfo}_s(i)$ without its terminal character.
\end{preuve}

Note that a proper prefix of $\mbox{lfo}_s(i)$ might label several
paths from the root in $PT(s).$

Let $[i,j]$ be an interval on $s=s_1..s_{n}$ and let
$\mbox{Support}([i,j])$ be the minimal position $p, i \leq p \leq j$, of the rightmost
occurrences of each letter in $s_i \ldots s_j.$ We define $O_s^{[i,j]}$ as
$\mbox{fo}_s(\mbox{Support}([i,j]),j).$ For instance, if $s = a_1 b_2
a_3 c_4 e_5 a_6 b_7 a_8 c_9 d_{10}\#_{11}$, $\mbox{Support}([1,3])=2,$ $\mbox{Support}([4,10])=5,$ $O_s^{[1,3]}=ba$ and $O_s^{[4,10]}=eabcd.$


\begin{definition}
Let $s= s_1..s_n$ and $1 \leq i \leq j \leq n$. We define
$\mbox{Extend}_s(i,j)$ as the maximal location reached when extending
the interval $[i,j]$ to the left and to the right while the closest
external characters $s_{i-1}$ or $s_{j+1}$ (if they exist) belong to
$C_s(i,j)$.
\end{definition}

 For instance, if $s = a_1\; b_2\; a_3\; c_4\; e_5\; a_6\; b_7\; a_8\;
c_9\; d_{10} \#_{11}$, $\langle 1,4 \rangle = \mbox{Extend}_s(2,4)$ and
$\langle 1,9 \rangle = \mbox{Extend}_s(2,7)$

\begin{lemma}
Let $\langle i,j \rangle$ be a maximal location of $s=s_1..s_n$.  There exists
a permutation of all characters of $C_s(i,j)$ that labels a path from the root in
$PT(s).$
\label{lemmain}
\end{lemma}
\begin{preuve}
$O_s^{ \langle i,j \rangle}$ is obviously a permutation of $C_s(i,j)$
  and a proper prefix of $\mbox{lfo}_s(\mbox{Support}( \langle i,j
  \rangle))$, which, by lemma \ref{alllfo}, labels a path from the
  root in $PT(s).$
\end{preuve}

\begin{corollary}
Let $s= s_1..s_n$. For all $i,j, 1 \leq i \leq j \leq n$, there exists
a permutation of all characters of $C_s(i,j)$ that labels a path from the root in
$PT(s).$ \label{allfingerprintstree}
\end{corollary}
\begin{preuve}
It suffices to extend the segment $s_i..s_j$ to $ \langle k,l
\rangle=\mbox{Extend}_s(i,j)$ in which it is contained. Then
$C_s(i,j)=C_s(k,l)$ and lemma \ref{lemmain} applies.
\end{preuve}

Let $z= ((r,\alpha_1,p_1), \ldots ,(p_{i-1},\alpha_{i},p_i))$ be a
path in $PT(s= s_1..s_n)$ from its root $r$. By notation extension, we
denote $\mbox{Suff}(z) = \mbox{Suff}(p_i)$.  Let $\mbox{SPref}(s)$ be
the set of all such paths and $w(z) = \alpha_1\alpha_2..\alpha_i$. Let
${\cal P(L)}$ be the set of all sets of maximal locations.We consider
the function $\Phi$ formally defined as:
$$\begin{array}{l|ccl}
\Phi : & \; \mbox{SPref}(s) & \longrightarrow &   {\cal P(L)} \\
    & \; z & \longmapsto &  \{  \langle k,l \rangle \in {\cal L}\; | \; O_s^{ \langle k,l \rangle} = w(z) \mbox{ and } 
\mbox{Support}( \langle k,l \rangle) \in \mbox{Suff}(z) \}
\end{array}$$

\begin{lemma}
Let $z= ((r,\alpha_1,p_1), \ldots ,(p_{i-1},\alpha_{i},p_i))$ be a
non-empty path in $\mbox{SPref}(s).$ Then $\Phi(z)\not=\emptyset$.
\label{phinonempty}
\end{lemma}
\begin{preuve}
By construction of the participation tree, there exits $m \in
\mbox{Suff}(z)$ such that $\alpha_1 \ldots \alpha_{i}$ is a proper prefix of
$\mbox{lfo}(m)$. Let $p$ be the first position of $\alpha_i$ in $s$
following $m.$  Then $\cup_{1 \leq f \leq
i} \{ \alpha_f\}$ $=$ $C_s(m,p).$ Let $ \langle k,l \rangle =
\mbox{Extend}_s(m,p)$. 

We prove now that $\mbox{Support}(\langle k,l \rangle)=m$. As
$\alpha_1 \ldots \alpha_{i}$ is a proper prefix of $\mbox{lfo}(m)$,
there exists an $\alpha=\mbox{lfo}(m)_{i+1}$ such that there is no
occurrence of $\alpha$ in the interval $[m,p]$, and thus after the
extension of $[m,p]$ to a maximal location $\langle k,l \rangle$, the
indice $l$ is strictly less than the indice of the first occurrence of
$\alpha$ after $m$. As, by definition of $\mbox{lfo}(m)$, there is no
occurrence of $s_m$ before the indice of $\alpha$ after $m$ in $s$,
there is no other occurrence of $s_m$ at the right of $s_m$ in the
interval $[m,l].$ Moreover, since all characters in $\alpha_1 \ldots
\alpha_{i}$ and only them appear after $m$ in $[m,l]$ in the order of
$\alpha_1 \ldots \alpha_{i}$ and the extension procedure ensures
that all characters in $[k,m]$ are characters from $\alpha_1 \ldots
\alpha_{i}$, we have $\mbox{Support}(\langle k,l \rangle)=m.$

Finally, it is obvious that $O_s^{ \langle k,l \rangle}= O_s^{[m,p]} =
\alpha_1..\alpha_{i}= w(z)$, and thus $\langle k,l \rangle \in
\Phi(z).$
\end{preuve}

\begin{lemma}
Let $z_1,\, z_2 \in \mbox{SPref}(s)$ be two distinct non-empty paths.  
Then $\Phi(z_1) \cap \Phi(z_2) = \emptyset$.
\label{phinonover}
\end{lemma}
\begin{preuve}
Assume {\em a contrario} that there exists $ \langle k,l \rangle \in
\Phi(z_1) \cap \Phi(z_2)$. Let $m=\mbox{Support}(\langle k,l
\rangle),$ $m \in \mbox{Suff}(z_1)$ and $m \in \mbox{Suff}(z_2)$. Thus
one of the paths is a prefix of the other. As $O_s^{[k,l]}= w(z_1) =
w(z_2)$, the two paths must be equal, which contradicts the hypothesis.
\end{preuve}

\begin{lemma}
Let $ \langle i,j \rangle$ and $ \langle k,l \rangle$ be two distinct
maximal locations of $s=s_1..s_n$ in the same equivalence class of
${\cal L}_C$. Then there exits $z \in \mbox{SPref}(s)$ such that 
both $ \langle i,j \rangle $ and $ \langle k,l \rangle $ are
contained in $\Phi(z)$.
\label{inthesame}
\end{lemma}
\begin{preuve}
 Let $m_1=\mbox{Support}(\langle i,j \rangle)$ and
 $m_2=\mbox{Support}(\langle k,l \rangle)$.  As $s_i..s_j=s_k..s_l$, $u=
 s_{m_1}..s_j$ $=$ $s_{m_2}..s_l$ and $m_1$ and $m_2$ are thus in the
 subtree of the path $h$ labeled by $u$ in $ST(s).$ After reduction of
 this path in $PT(s)$, the resulting path $z$ is such that $w(z)=O_s^{
 \langle i,j \rangle}= O_s^{\langle k,l \rangle}$, so $m_1, m_2 \in
 \mbox{Suff}(z).$ Thus $\langle i,j \rangle, \langle
 k,l \rangle \in \Phi(z).$
\end{preuve}

\begin{theorem}
Any maximal location is contained in the image $\Phi(z)$ of some path $z$ in
$PT(s=s_1..s_n)$, and the size of $PT(s)$ (without the initial positions of suffixes)
is $O(|{\cal L}_{C}|)$.
\end{theorem}
\begin{preuve}
Lemma \ref{lemmain} directly implies that all maximal locations are in the
image $\Phi(z)$ of a path $z$ in $PT(s)$. As by lemma \ref{phinonover}
the images $\Phi(z)$ are non-overlapping, they form a partition of
${\cal L}$. Lemma \ref{inthesame} ensures that ${\cal L}_C$ partition
is a subpartition of the partition formed by the images of $\Phi$. As by
lemma \ref{phinonempty} there is no empty image, the number of such
images is smaller than or equal to $|{\cal L}_C|.$
\end{preuve}

Note that we considered the size of $PT(s=s_1..s_n)$ without the
initial positions of suffixes (square boxes in Figure
\ref{parttree}). With these positions, the size of $PT(s)$ 
is $O(n+|{\cal L}_{C}|)$.

We explain below how to compute the participation tree from the
suffix tree in linear time.

\subsubsection{From Suffix Tree to Participation Tree}
\label{parttreeconstruct}

We extend the notion of $\mbox{fo}_s(i,j)$ keeping the positions of
the characters in $s=s_1 .. s_n$.  We define $\mbox{efo}_s(i)$ as the
string formed by concatenating the first occurrences of each distinct
character touched when reading $s$ from position $i$ (included) to
position $n$ (included) but indexed by the position of this character
in the sequence.  For instance, if
$s = a_1b_2a_3c_4e_5a_6b_7a_8c_9d_{10}\#_{11}$,
$\mbox{efo}_s(3)= a_3c_4e_5b_7d_{10}\#_{11}$ and $\mbox{efo}_s(5) =
e_5a_6b_7c_0d_{10}\#_{11}.$

The idea of the algorithm is the following. For each transition
$(i,j)$ on the path of a longest suffix $v=s_k\ldots s_n$, we compute
the ``participation'' of the edge to $\mbox{lfo}_s(k)$, that is, the
new characters the edge brings in $\mbox{lfo}_s(k)$. For
instance, in Figure \ref{suffixtree} the participation of edge
$(6,8)=[5,11]$ is $e$, since it is on the path of the suffix
$s_3 \ldots s_n$ and $\mbox{lfo}_s(3) = ace.$ The participation of
edge $(12,14)=[5,11]$ is $eab$ since $\mbox{lfo}_s(4) = ceab.$ 

To compute the participation of interval $[i,j]$ on the path of a
 suffix $v=s_k\ldots s_n$, we use $\mbox{efo}_s(k)$ and also
the next position of $s_k$ after $k$ in $s$, if it exists. Assume it
is the case and let $p$ be this position. Thus $s_p=s_k$. Let
$\mbox{efo}_s(k)= s_k s_{l_1} s_{l_2} \ldots s_{l_z}$and $ l_h\leq p
\leq l_{h+1}$. If $i\geq p$, the participation of $[i,j]$ is the empty
word $\varepsilon$. Otherwise, if $i<p$ then the participation of $[i,j]$
is the string (potentially empty) $s_{l_a}\ldots s_{l_b}$ with
\begin{itemize}
\item $i \leq l_a$ and $l_a$ is the smallest such indice;
\item $l_b \leq min(j,p-1)$ and $l_b$ is the greatest such indice.  
\end{itemize}

[Note that this computation requires that the interval $[i,j]$ which
annotates a transition in the suffix tree corresponds to the suffix
$v$ used as reference. In order to ensure this, below we "shift" each
interval $[i,j]$ according to the suffix we are currently reading
before computing its participation.]



\begin{figure}[tb]
\begin{center}
\begin{algorithme}
{\sc Build\_part\_tree}($ST(s=s_1..s_ns_{n+1}$ with $s_{n+1} = \#)$)\\
\lign $\mbox{efo}_s(n)=s_n$ and $p_n = n+1$\\
\lign \For{$i = n..1$}\\
\lign  \> $\mbox{length} \leftarrow n$\\
\lign  \> $\mbox{Current} \leftarrow \mbox{Leaf}(i)$ in $ST(s)$.\\
\lign  \> \While{$\mbox{Current}$ not marked \OurAnd $\mbox{Current}\neq \mbox{Root}$}\\
\lign  \> \>  $\mbox{Prec} \leftarrow \mbox{Parent}(\mbox{Current})$ in $ST(s)$.\\
\lign  \> \> $[k,l] \leftarrow \mbox{ edge } (\mbox{Prec},\mbox{Current})$\\
\lign  \> \> $[\mbox{pos\_deb},\mbox{pos\_end}] \leftarrow [\mbox{length}-(l-k),\mbox{length}]$\\
\lign  \> \> Compute the participation of $[\mbox{pos\_deb},\mbox{pos\_end}]$ in  $\mbox{efo}_s(i)$\\
\lign  \> \> Mark $\mbox{Current}$\\ 
\lign  \> \> $\mbox{length} \leftarrow \mbox{length}-(l-k)-1$\\
\lign  \> \Endwhile\\
\lign  \> $\mbox{efo}_s(i-1)$ $\leftarrow$ Update  $\mbox{efo}_s(i)$\\
\lign \Endfor\\
\lign Replace each terminal character of all paths from the root by $\varepsilon$.\\
\lign Remove $\varepsilon$ edges by node merging.
\end{algorithme}
\end{center}
\vspace{-0.5cm}
\caption{Building the participation tree from the suffix tree.}
\label{algo:parttree}
\vspace*{-0.5cm}
\end{figure}


For instance, in Figure \ref{suffixtree}, $\mbox{efo}_s(2)= b_2 a_3
c_4 e_5 d_{10} \#_{11}$ and $p=7$ since $7$ is the next position of
$b$ after position $2$. Thus, participation of edge $(1,9)=[2,4] = b_2
a_3 c_4 = bac$, participation of $(9,11)=[5,11] = e_5 = e$ (since
$p=7$). For each suffix $[k,n]$, given $\mbox{efo}_s(k)$ and $p$, a
bottom-up process from leaf $k$ to the root of the suffix tree allows
us to: 
\begin{itemize}
\item[(a)] shift the pointed positions to positions corresponding to the suffix
considered. The bottom-up approach allows to read the suffix from its
end, and thus the sizes of the encountered transitions are enough to
know which segment of the suffix the edge represents;
\item[(b)] compute the participation of each (not previously touched)
edge on this path. Also, the bottom-up approach allows us to avoid
unnecessary computation, since the participation of an upper edge ends
in $\mbox{efo}_s(k)$ where the participation of the lower begins.
\end{itemize}

We modify the suffix tree using successive $\mbox{efo}_s(k)$, for
$k=n..1$. A sketch of this algorithm is given in Figure
\ref{algo:parttree}. At the end of this process, we first replace the
terminal character of all paths from the root by $\varepsilon.$ We finally
remove all $(u,\varepsilon,v)$ edges by merging $u$ and $v.$

\begin{theorem}
\label{theo:build_part_tree}
The participation tree of $s=s_1..s_n$ can be built in $O(n+|{\mathcal
L}_{C}|)$ time and $O((n+|{\mathcal L}_{C}|)\log n)$ bits of space.
\end{theorem}
\begin{preuve}
The algorithm is correct since it consists of the direct
computation of the participation of each edge one after the other. We now
study its complexity.

For each suffix $[k,n]$, given $\mbox{efo}_s(k)$  and
$p$, the bottom-up process from leaf $k$ to the root of the suffix
tree can be done in $O(1)$ time for each unmarked node. 

We maintain each $\mbox{efo}_s(i)$ as a combination of a doubly linked
list and an array of size $\Sigma$ in which each cell $j$ points to
the position of character $f^{-1}_{\Sigma}(j)$ in the doubly linked
list. Thus, adding a character $c$ to the head of the doubly linked
list while recording its position in the corresponding cell of the
array is $O(1)$. Removing a character out of the list is also $O(1)$
since it suffices to find its position in the list using the array and
remove the character using the pointer to the previous and
next character in the list. Initializing the structure is $O(\sigma)$
but it has only to be done once. In addition to the array and the doubly
linked list, a pointer $tp$ points to the character in the list whose
position is just before $p$ (the next position of $s_i$ in $s$) if
such character exists or to the end of the list otherwise. An instance of this
structure is given in Figure \ref{efos}.

\begin{figure}[htb]
  \centering
 \includegraphics[width=6cm]{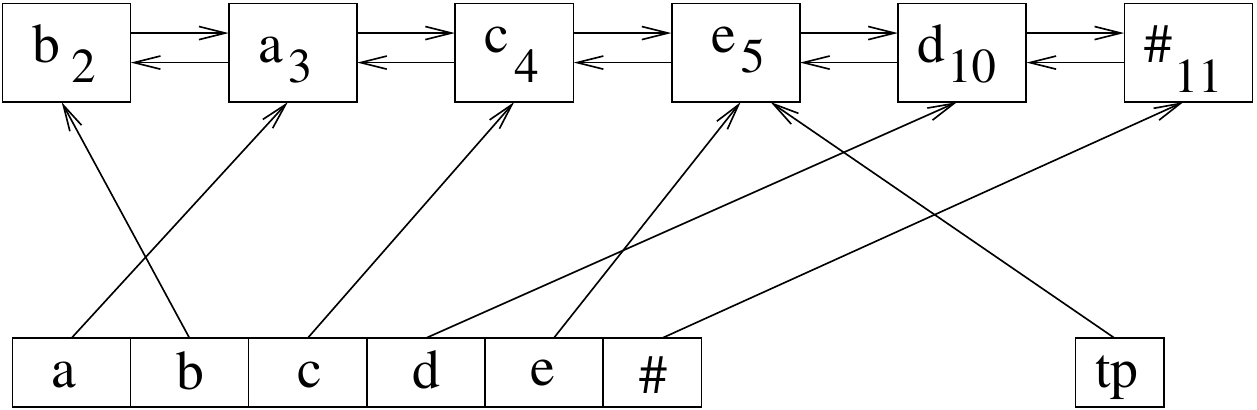}
 \caption{Data structure for maintaining $\mbox{efo}(i)$ shown on
$\mbox{efo}_s(2)= b_2 a_3 c_4 e_5 d_{10} \#_{11}.$ The pointer $tp$
points to the character in the list whose position is the largest
smaller position in the list compared to the next position $p$ of $b$ in
$s$, which is $7.$}
  \label{efos}
 \end{figure}

Assume that $\mbox{efo}_s(i)$ is represented in this way, with knowing $tp_i$,
the next position in the doubly linked list of the first character $s_{i-1}$
in $\mbox{efo}_s(i)$. To compute $\mbox{efo}_s(i-1)$ and $tp_{i-1}$, it
suffices to test in the array if $\alpha = s_{i-1}$ already appears in the
list. If yes, $tp_{i-1}$ points to the character just before $\alpha$ in the list,
if not $tp_{i-1}$ is set to the end of the list. Then $\alpha$ is
removed out of the list and inserted at its head. The first
$\mbox{efo}_s(n)$ is simply $s_n$, and $tp$ points to the end of the
list.

Computing the participation of each non-touched edge on a path
from the root to a leaf corresponding to suffix $i$ in a bottom-up
manner is not expensive since it suffices to ``consume''
$\mbox{efo}_s(i)$ backward from $tp_i$ edge after edge as soon as an
edge $[k,l]$ (shifted to correspond to suffix $i$) is such that $k$ is
less than the position of the element pointed by $tp_i$.  Thus,
calculating the participation of each edge in the suffix tree can be
done in a time proportional to the participation of the edge in
$PT(s)$ tree plus the total number of edges in the tree.

Replacing the terminal character of each path from the root by
$\varepsilon$ is $O(n).$ Merging each of the $\varepsilon$ edges can
also be performed in $O(n)$ since each such $\varepsilon$ edge is
either a previous edge of the suffix tree or was labeled by a single
terminal character of a path from the root. The whole construction of
$PT(s)$ is thus $O(n + |{\mathcal L}_{C}|)$ time.

The space required is the size of the suffix tree plus the size of the
participation tree plus the size of the data structure representing
$\mbox{efo}_s(i)$, thus $O(n+|{\mathcal L}_{C}|)$ space.
\end{preuve}

\vspace*{0.5cm}

\noindent
We now explain how to name all fingerprints from the
participation tree. 

\subsubsection{Naming a Participation Tree}
\label{newnaming}


\begin{figure}[tb]
\begin{center}
\begin{algorithme}
{\sc Depth\_first\_search}($FT_k$,$Current$)\\
\lign \For{all $\alpha$ such that $\delta(Current,\alpha)\not= \Theta$}\\ 
\lign  \> $ q \leftarrow \delta(Current,\alpha)$\\
\lign  \>  $\{[a],j\} \leftarrow \Delta(Current,\alpha,q)$\\
\lign  \>  $prec \leftarrow FT_k[j]$\\
\lign   \> $FT_k[j] \leftarrow [a]$\\
\lign   \> $\Delta(Current,\alpha,q) \leftarrow \{(FT_k[2 \lfloor j/2 \rfloor], FT_k[2
  \lfloor j/2 \rfloor+1]),\lfloor j/2\rfloor \}$\\
\lign   \> {\sc Depth\_first\_search}($FT_k$,$q$)\\
\lign   \>  $FT_k[j] \leftarrow prec$\\
\lign  \Endfor\\\\
{\sc Name\_fingerprint}($PT(s)$)\\
\lign $ninit_1 \leftarrow [0]$\\
\lign \For{$k = 1..\log \sigma$}\\
\lign  \>    $FT_k \leftarrow$ name table of size $\sigma/2^{k-1}$
all initialized to $ninit_k$\\
\lign \> {\sc Depth\_first\_search}($FT_k$,$\mbox{Root}(PT(s))$)\\
\lign \> $Sl \leftarrow \Theta$~~~ /* empty stack */\\
\lign \> \For{ all edges $e=(p,\alpha,q)$ in $PT(s)$}\\
\lign \>   \> $\{(n_1,n_2),j\} \leftarrow \Delta(p,\alpha,q)$\\ 
\lign  \>  \> Add $(n_1,n_2)$ to $Sl$.\\
\lign  \> \Endfor\\
\lign  \>  add the couple $(ninit_k,ninit_k)$ to $Sl$\\
\lign  \>  sort $Sl$ in lexicographical order\\
\lign  \>  give new names for each different couple in $Sl$\\
\lign   \> replacing each pair in $\Delta(p,\alpha,q)$ by its new name\\
\lign  \>   $ninit_{k+1} \leftarrow $ name of the pair $(ninit_k,ninit_k)$\\
\lign \Endfor
\end{algorithme}
\end{center}
\vspace{-0.5cm}
\caption{Naming all fingerprints in a participation tree $PT(s)$.}
\vspace*{-0.5cm}
\label{algo:namealltwin}
\end{figure}

The naming approach of the previous section has been modified in
\cite{Ours2006} to name on the same set of names a table of lists of
fingerprint changes. The main modification is that the linear sorting
is done for each level on all the pairs of all the lists of the table.
We use a similar approach, but instead of a table of lists we consider
the set of all paths from the root in the participation tree
$PT(s)$. Each such path is considered as a list of fingerprint
changes, except that the initialization of the naming list is done
once for all paths. Corollary \ref{allfingerprintstree} guarantees
our approach. The {\sc Name\_fingerprint} algorithm names all
fingerprints. Its pseudo-code is given in Figure
\ref{algo:namealltwin}.

As in the list naming of section \ref{subchanges}, $\log \sigma$
iterations are performed, one by fingerprint array level (loop 11-24),
the lowest one excepted. With each edge $(p,\alpha,q)$ of $PT(s)$ a value
$\Delta(p,\alpha,q)$ is associated. At the end of iteration $k$,
this value records the change corresponding to the edge in the
fingerprint array of level $k+1$. The value $\Delta(p,\alpha,q)$ is
assumed to be initialized with $\{[1],f_{\Sigma}(\alpha)\}$
corresponding to the change induced by the edge at the lowest level
$1$.

In each iteration $k$, the recursive algorithm {\sc
  Depth\_first\_search} is called (line 13) on the participation tree
to update all values $\Delta(p,\alpha,q)$ during a depth first
search. The update operation on each such value is similar to the pair
update in the naming of a simple list of fingerprint changes in
section \ref{subchanges}. Note that in {\sc Depth\_first\_search} a
special $FT$ table is modified (line 5) before the recursive call but
reinitialized to the previous value after the call (line 8). This
permits to initialize the table $FT$ only once before the first call
to {\sc Depth\_first\_search} (line 12) and thus the initialization
costs are the same for all paths as for a single list, and thus are
bounded by $2\sigma$.

After the depth first search the values $\Delta(p,\alpha,q)$ are
collected on all the edges $(p,\alpha,q)$ of the participation tree (lines
14-18) in a list $Sl$. This list is lexicographically sorted and a new name is given
to each unique pair (line 20), similarly to the naming of a single
list in section \ref{subchanges}. The initial pair of names of each
$\Delta(p,\alpha,q)$ is then replaced by its new name.

To initialize the fingerprint array at the next level, the couple
$(ninit_k,ninit_k)$ is added to the list of names (line 19) and its new
name is retrieved after the sorting and the renaming (line 22).

\begin{theorem}
\label{theo:build_names}
The {\sc Name\_fingerprint} algorithm applied on $PT(s)$ names all
fingerprints of $s$ in $\Theta(\sigma+|{\cal L}_{C}|\log\sigma)$ time 
using $O((|{\cal L}_{C}|+|{\cal F}|\log\sigma)\log n)$ bits of working space.
\end{theorem}

\section{A Space Efficient Fingerprint Representation}
\label{space}
\subsection{Overview}
In this section we show how the fingerprint set can be represented
in just $|{\cal F}|(2\log \sigma+\log_2 e)(1+o(1))$ bits of space instead of
$O(|{\cal F}|\log n)$ bits. Our solution is particularly attractive whenever
$\sigma$ is sufficiently small (e.g. $\log \sigma=o(\log n)$) as it saves a factor
$\Theta(\frac{\log n}{\log\sigma})$ compared with a standard non-succinct representation
that uses at least $\Theta(|{\cal F}|)$ words of space, which translates into 
$\Theta(|{\cal F}|\log n)$ bits. 

Our representation relies on the fingerprint trie as described in section \ref{subsubsec:finger_trie}.

Before describing our solution, we first recall some basic facts on the fingerprint trie that will be needed to understand our solution. First, recall the following two facts:
\begin{enumerate}
\item Each node in the trie corresponds to a unique set and each set
corresponds to a unique node.
\item Each prefix of a fingerprint is also a fingerprint. 
\end{enumerate}

Note also that the fingerprint trie implies an ordering on the 
characters of any given fingerprint represented in the trie.
More precisely for a given node $q$,
the characters of the corresponding fingerprint $f_q$ are
ordered according to the order in which they appear as labels of the nodes 
in the path from the root to the node $q$. 

In our representation, the fingerprint trie will be represented in two different ways. This is why the space usage will be at least $2|{\cal F}|\log\sigma$ bits. The first representation will permit a traversal of the fingerprint trie bottom-up (climb the trie) and the second one will permit a traversal of the fingerprint trie top-down (descend the trie). If the fingerprint is represented in the trie, then a bottom-up traversal will permit one to get the proper ordering on the fingerprint characters. Then, the presence of the fingerprint can be confirmed by a top-down traversal. Note that this second traversal can only return true if the fingerprint exists and is in the correct order represented in the trie.  Therefore a top-down traversal will never return a false positive answer (it will never return true for a fingerprint not represented in the trie or for fingerprint represented in the trie but with a different ordering). Likewise, this top-down traversal will never return a false negative (it will always give a positive answer for an existing fingerprint) as it will be proven later that a bottom-up traversal will always return the correct ordering of the characters of an existing fingerprint and this correct ordering will thus be used to do a successful top-down traversal of the trie.  

We now give more details on our representation. First, notice that each set (fingerprint) uniquely corresponds to a distinct node of the fingerprint trie. Let $f_q$ denote the fingerprint associated with the node $q$. Let $\alpha(q_1,q_2)$ denote the characters that label the edge which connects a node $q_1$ to its child $q_2$. Notice that by definition of the fingerprint trie for any node $q_2$ having a parent $q_1$, we have $f_{q_2}=f_{q_1}\cup \{\alpha(q_1,q_2)\}$. That is, the fingerprint of the node $q_2$ is obtained by adding one character $\alpha(q_1,q_2)$ to the fingerprint of its parent node $q_1$. 
\\
The solutions we propose are able to find whether a given query fingerprint $f$ is in the set ${\cal F}$ in $O(|f|)$ time. A query for a fingerprint $f$ represented by a string which contains all the characters of $f$ in an arbitrary order will work in three steps:
\begin{enumerate}
\item We query the bottom-up representation of the trie, which, when given the fingerprint $f$, returns a string $s$ of length $|f|$. This bottom-up representation relies on the use of succinct function representation of lemma~\ref{lemma:succ_fn_rep}. A detailed description of the step is in section~\ref{subsec:backtrack_fn}.

\item We check whether the string $s$ is a permutation of the set $f$. That is, we check whether $s[i]\in f$ for each $i\in [0,|f|-1]$ and check also that all characters of $s$ are distinct. This step is done in time $O(|f|)$ with high probability using $O(|f|\log\sigma)$ bits working space or in deterministic time $O(\epsilon|f|)$ using working space $O(\sigma^{1/\epsilon}\log\sigma)$ bits for any positive integer $\epsilon$. A detailed description of the step is in section~\ref{subsec:succtrie}.
\item The final step is using the succinct top-down representation of the trie to do a top-down traversal for the string $s$. This step permits checking whether the string $s$ exists in the trie representation in $O(|s|)=O(|f|)$ time. Notice that this is equivalent to checking that $f\in {\cal F}$. This is the case as by previous step we have checked that $s$ is a permutation of $f$ and we know that the trie stores a unique string corresponding to each fingerprint in ${\cal F}$. A detailed description of the step is in section~\ref{subsec:equal_test}. 
\end{enumerate}

In the following three subsections we describe in more detail the data structures used for each
of the three steps. In subsection~\ref{subsec:fullquery} we give the full picture of the query
and prove its correctness. 

\subsection{Backtracking Function (bottom-up trie representation)}
\label{subsec:backtrack_fn}

The first step is achieved through a data structure we call the
 \emph{backtracking function}, which is in fact a bottom-up representation
 of the trie. This function associates to each fingerprint
 $f_i$ the last character in its string representation $s_i$. 
We will simply use a static function that maps each
set (fingerprint) to the last character in the character ordering.
In other words whenever we have a fingerprint $f$ corresponding to a node $q$ 
in the fingerprint trie, we associate with $f$ the character which labels the edge 
which connects $q$ to its parent in the trie.
That is, for each set we have a string representation
that contains exactly the same characters as the set in a certain order. With each set we associate the last character in its
string representation. 

 It turns  out that representing this backtracking function can be done using
 just $(|{\cal F}|\log \sigma)(1+o(1)))$ bits of space which is optimal. The
 generation of the backtracking function from the set ${\cal F}$ can be done
 in optimal $O(|{\cal F}|)$ time. The generation is based on the
 use of a polynomial hash function (the same used in the so-called Rabin-Karp
fingerprints~\cite{KR87}). The first step consists in a
 top-down traversal of the fingerprint trie. Recall that each node
 represents a distinct fingerprint. Given a node $q$ with a parent
 $p$, we note the fingerprint associated with $p$ by $f_p$ and the
 fingerprint associated with $q$ by $f_q$. Then, if the edge which
 connects $p$ to $q$ is labeled by character $\alpha$, we will have
 $f_q=f_p\cup \{\alpha\}$. So, during the top-down traversal of the trie we
 will compute a hash value associated with each fingerprint. For
 that we will make use of the polynomial hash functions family as
 described in section~\ref{subsec:poly_hash}. More precisely, the hash
 functions we will use are polynomials modulo a prime $P$ chosen such
 that $P\in[|{\cal F}|^2\sigma,2|{\cal F}|^2\sigma]$. Finding $P$ takes 
time $O((\log(|{\cal F}|^2\sigma))^c)=O((\log(|{\cal F}|+\log\sigma))^c)$ for some constant $c$. 
(see ~\ref{subsec:poly_hash} for details on the algorithm used to find $P$). 

Before beginning the top-down traversal of the
 trie, we will randomly choose a number $r$ from the interval
 $[0,P-1]$. For any fingerprint $f_i$ having elements
 $\alpha_1,\alpha_1,\ldots,\alpha_{|f|}$, we will associate the
 hash value computed using the formula
 $H(f_i)=r^{f_{\Sigma}(\alpha_1)}+r^{f_{\Sigma}(\alpha_2)}+\ldots+r^{f_{\Sigma}(|f|)}$ where multiplications and
 additions are all done modulo $P$.  \\Now the generation of the hash
 values for all fingerprints is done in the following way: we first
 associate the hash value $0$ with the root node which does not
 represent any fingerprint. We note by $H_q$ the hash value associated
 with the node $q$ and by $H_p$ the hash value associated with node
 $p$. From the definition it is evident that $H_q=H_p+r^{f_{\Sigma}(\alpha)}$ where $\alpha$
 is the character which labels the edge connecting node $p$ to node
 $q$. Therefore, during a top-down traversal of the trie, we can compute the
 hash value for each fingerprint in constant time given the
 fingerprint of its parent node.  Once we have generated the $|{\cal F}|$
 hash values corresponding to the $|{\cal F}|$ fingerprints, we will check
 whether all fingerprints are distinct. According to lemma~\ref{lemma:poly_hash_lemma}
we deduce  that this is the case with probability of at least $1/2$. 
 If this is not the case, we will choose a
 new value $r$ and recompute the hash values in the same way during a
 top-down traversal of the trie. As on expectation we will do $O(1)$
 trials and each trial taking time $O(|{\cal F}|)$, we deduce that 
the total expected time is $O(|{\cal F}|)$.
 \\ Once we have successfully mapped
 all the keys to distinct hash values in range $[0,P-1]$, we will
 store a static function using lemma~\ref{lemma:succ_fn_rep} which for each fingerprint $f_i$ will
 associate the character $f_{\Sigma}(\alpha_i)$ (where $\alpha_i$ is the last character in
 $f_i$) to the hash value $H(f_i)$. The space used by the static function will clearly be $|{\cal F}|(\log \sigma)(1+o(1))$ bits.
\subsection{Deterministic and Probabilistic Set Equality Testing}
\label{subsec:equal_test}
We now describe a method to test for set equality. This is step 2 in our 
query algorithm. Given two strings $s_1$ and $s_2$
where $|s_1|=|s_2|$, we would wish to test whether the two strings are
permutations of the same set\footnote{To declare that two strings are equal we require that 
the two strings are permutations. That is, the characters of each string are all distinct.}. 
That is, we are asking if we can obtain
the string $s_1$ by doing a permutation on the characters of the
string $s_2$. We propose two solutions for this problem. The first one
is randomized while the second one is deterministic. The two solutions are folklore, 
but we describe them here for completeness. 

\paragraph{Randomized Method}

The randomized method works in the following way : we use a dynamic
perfect hash table~\cite{DKMHRT88}(or any other efficient hash table 
implementation) in which we insert all the characters of the string
$s_1$. This takes time $O(|s_1|)$ with high probability and uses space $O(|s_1|\log\sigma)$ bits
\footnote{A linear-space hash table needs $O(\log |U|)$ per element where $U$ is the universe.
In our case $U=\Sigma$ and thus $|U|=\sigma$.}. 

During the insertion, we can easily check that the characters of $|s_1|$ are all distinct by checking that 
every character of $s_1$ is not present in the table at the time of its insertion. 
In the hash table, we associate a bit with each key and we initialize the bit to
zero. Now, we process the string $s_2$. For each character $\alpha$ of
$s_2$ we query the perfect hash table for the character $\alpha$. In case
we do find it, we mark the bit associated with it. After we have
processed all characters of $s_2$, we check if all the bits
associated with characters of $s_1$ are now set to one. If this is the
case, we conclude that $s_2$ and $s_1$ are permutations of the same
set.

Clearly this randomized method uses $O(|s_1|)$ words of space that is
$O(|s_1|\log\sigma)$ bits of space, which is optimal up to a
constant-factor, as we also need $|s_1|\log\sigma$ bits to represent
$|s_1|$.

\paragraph{Deterministic Method}

We now describe a deterministic method which can be used to do equality testing. The basic method needs
$\sigma$ bits of working space for queries and checks set equality
in optimal time $O(|s_1|)$. A more sophisticated method could use space
$O(\sigma^{1/k}\log\sigma+|s_1|\log\sigma)$ bits and answers set equality in time $O(k|s_1|)$ for
any integer $k$ such that $k>1$.  In the basic method, we will simply
use a bitvector $B$ of $\sigma$ bits. At the beginning all the bits in $B$ 
are set to zero, and we require that they are reset to zero after each equality 
test. 

The equality test works in the following way: we first process the string
$s_1$. For each $i$ in $[0,|s|-1]$, we set $c=f_{\Sigma}(s_1[i])$ and then set $B[c]=1$.
Before setting $B[c]=1$, we check that $B[c]\neq 0$ and thus that the character $s_1[i]$ 
does not occur twice in $s_1$. 

We now traverse the string $s_2$. For each $i$ in $[0,|s|-1]$, 
we set $c=f_{\Sigma}(s_1[i])$ and check that $B[c]=1$. If this was the case, then we set $B[c]=0$, 
otherwise, we declare that $s_1$ and $s_2$ are two distinct strings. Setting $B[c]$ to zero
is necessary to ensure that all the characters of $s_2$ are all distinct. 

It is easy to see that the above procedure correctly computes the equality of $s_1$ and 
$s_2$. In the first phase we have set all the $|s_1|$ distinct bits corresponding to
characters of $s_1$. In the second phase, we check that the bits corresponding to characters
of $s_2$ are all distinct and all set which can only be the case if those bits are precisely
the $|s_1|$ bits corresponding to character of $|s_1|$.

At the end of checking, if the two strings are equal, then all the bits of $B$ are set to zero, so that 
$B$ is ready for the next query. If the two strings are not equal, then we need to 
traverse the string $s_1$ and clear the bits of $B$ which were 
set to one when $s_1$ was first traversed (we set $B[c]=0$ for every $c=f_{\Sigma}(s_1[i])$)

\begin{lemma}
\label{lemma:simple_det_equ_test}
We can do equality testing between two strings $s_1$ and $s_2$ over an alphabet of size $\sigma$
in time $O(|s_1|)$ using $\sigma$ bits of working space. 
\end{lemma}

We now describe the more sophisticated method. We only 
describe how to achieve $O(\sqrt{\sigma})$ space. The generalization 
to $O(\sigma^{1/k})$ space for $k>2$ can easily be deduced from the 
case $k=2$. 

The method works in the following way: we first partition the characters of $s_1$
according to their $\lceil \log\sigma/2\rceil$ most significant bits.  
We also do the same partitioning for the characters of $s_2$. 
Finally, we compare all the pairs of partitions (one from  $s_1$ and one from $s_2$) in which the characters
share the same $\lceil \log\sigma/2\rceil$ most significant bits. 

We now give the details of the implementation. 
We use a table $T_1$ with $2^{\lceil\log\sigma/2\rceil}\leq 2\sqrt{\sigma}$ cells where 
each cell $T_1[i]$ contains a pointer (denoted by $T_1[i].P$) to a list of characters. 
At the beginning we suppose that every $T_1[i].P$ is initialized to null meaning that 
all the lists are empty. We also use a list $L_1$ which stores a list of non-empty cells
(cells with non null pointers) of $T_1$. At the beginning we process the characters of $s_1$ one by one 
and for each character $\alpha_i$ do the following steps:
\begin{enumerate}
\item Compute $j=MSB(f_{\Sigma}(\alpha_i))$, the $\lceil \log\sigma/2\rceil$ most significant  bits of $f_{\Sigma}(\alpha_i)$. 
\item Save in variable $oldP$ the old value of $T_1[j].P$. 
\item Add $\alpha_i$ to the list $T_1[j].P$.
\item If $oldP$ equals null, add $j$ to the list $L_1$. That is, the list $T_1[j].P$
which was previously empty is added to $L_1$ as now it is non-empty. 
\end{enumerate}
At the end of the processing, we do a second step in which we use a second table $T_2$ similar to $T_1$, where each cell $T_2[i]$ 
has a field $Z_{T_2[i].P}$. In this step we process the characters of $s_2$ one by one 
and for each character $\alpha_i$, we add $\alpha_i$ to the list $T_2.P[j]$. 
In the third step, we use two lists $L'_1$ and $L'_2$ initially empty. We take the list $L_1$ and for each element $j$ in the list do the following: 
\begin{enumerate}
\item Add all elements of the list $T_1[j].P$ at the end of the list $L'_1$. 
\item Add all elements of the list $T_2[j].P$ at the end of the list $L'_2$.
\end{enumerate}
At the end of the third step we are left with two lists $L'_1$ and $L'_2$ which are sorted
according to the list $L_1$. That is, in each of the two lists we  have 
first all characters whose $\lceil \log\sigma/2\rceil$ most significant bits are equal to $L_1[0]$ followed 
by all characters whose most significant are equal to $L_1[1]$ etc.
Thus, to finish the equality testing it suffices for every $j$ in the list $L_1$ to do the following:
\begin{enumerate}
\item First advance in $L'_1$ in order to find $R_{j1}$ the longest run of $t_1$ characters in $L'_1$ whose $\lceil \log\sigma/2\rceil$ most significant bits are equal to  $j$. 
\item Similarly, advance in $L'_2$ to identify $R_{j2}$ the longest run of $t_2$ characters in $L'_2$ whose $\lceil \log\sigma/2\rceil$ most significant bits are equal to  $j$.
\item Check that $t_1=t_2$. If this is not the case, immediately declare that $s_1$ is distinct from $s_2$.
\item Otherwise  we check for the equality of the characters in $R_{j1}$ and $R_{j2}$. To this end we already know that they have the same $\lceil \log\sigma/2\rceil$ most significant bits, so that we only need to do equality testing for the $\lfloor \log\sigma/2\rfloor$ least significant bits between characters of 
$R_{j1}$ and $R_{j2}$, which can be done using the procedure of lemma~\ref{lemma:simple_det_equ_test}. This will take time $O(t_1)$ and needs to use just a bitvector of size $2^{\lfloor \log\sigma/2\rfloor}\leq \sqrt{\sigma}$ bits. 
\end{enumerate}
If all the iterations are completed, we immediately deduce that the two sets $s_1$ and $s_2$ are equal. 
Concerning the running time, it is clear that the above procedure runs in time $O(|s_1|)$. Every element of $L'_1$ and $L'_2$ is only traversed twice, the first time for determining the length of the runs and the second time for determining the equality between elements of two runs. Each time an element is traversed, only a constant number of operations are carried on. 

We now analyze the space usage. The total space needed to store the different lists will be upper bounded by $O(|s_1|\log\sigma)$. The table $T_1$ will use space $O(\sqrt{\sigma}\log\sigma)$ bits, while the bitvector $B$ will use space $O(\sqrt{\sigma})$ bits. 

The above algorithm can be easily generalized to use space $(\sigma^{1/k}\log\sigma)$. For that it suffices to do the partitioning of the characters of $s_1$ and $s_2$ in $k-1$ phases. 
The $\log\sigma$ bits of the characters are divided in slices of size about $\log\sigma/k$ bits each. Then in each phase we partition the keys according to a one of the slices starting from the most significant slice to the least significant. After $k-1$ partitioning we will be left with partitions which only differ in their (at most) $\log\sigma/k$ least significant bits. In the final phase, pairs of partitions (one from $s_1$ and one from $s_2$) can easily be matched as was done above using lemma~\ref{lemma:simple_det_equ_test}. 

\begin{lemma}
\label{lemma:query_step3}
Given any two strings $s_1$ and $s_2$ of equal length, testing for the equality of the multisets induced by $s_1$ and $s_2$ can be done: 
\begin{enumerate} 
\item In expected $O(|s_1|)$  time with high probability using only $O(|s_1|\log\sigma)$ bits of space.
\item In worst case $O(k|s_1|)$ time using $(\sigma^{1/k}\log\sigma)$ bits of space. 
\end{enumerate}
\end{lemma}

\subsection{Succinct Trie Representation (top-down trie representation)}
\label{subsec:succtrie}
The third step of a query uses a top-down trie representation which we describe in this section. 
First of all, a trie $Tr$ of size $N$ over an alphabet $\sigma$ can be represented compactly to use 
optimal space $N(\log\sigma+\log_2 e+o(1))$ using the representation described in~\cite{RRS07} 
permitting many navigation operations on the trie in constant time. In particular, a top-down traversal 
of the trie for a string $s$ can be done in time $O(|s|)$ by using $O(1)$ time at each step $i$ of the traversal 
which consists in finding the child labeled with character $s[i]$. 
Given a string $s$, we can determine whether $s\in S$ in time $O(|s|)$, by doing a top-down traversal of the trie. 
Thus, given the set ${\cal F}$ of fingerprints in a trie of size $|{\cal F}|$, we can succinctly encode 
the trie representing the set ${\cal F}$ in time $O(|{\cal F}|)$ so that the trie uses space of
$|{\cal F}|(\log \sigma)(1+\log_2 e+o(1))$ bits. A top-down traversal of the trie will take time $O(1)$ time per 
traversed node. Thus given a fingerprint $f$ in the correct order, we can check whether it is presented in the set~${\cal F}$ 
by doing a top-down traversal of the succinctly encoded trie representing the set ${\cal F}$.

\subsection{Putting Things Together}
\label{subsec:fullquery}
We are now ready to describe the full details of the queries on our data
structures described in the previous subsections.
A query for a fingerprint $f=\{\alpha_1,\alpha_2,\ldots,\alpha_{|f|}\}$ 
is given as a string $s_f$ of characters consisting in the concatenation of the characters 
$\alpha_1,\alpha_2,\ldots,\alpha_{|f|}$. The characters are not necessarily
lexicographically sorted. The query involves the following steps :
\begin{enumerate}
\item Compute the hash value: $$H(f)=\sum_{1\leq i\leq|f|}r^{f_{\Sigma}(\alpha_i)}$$

This operation takes time $O(|f|)$, as it involves only $O(|f|)$ arithmetic
operations. In the following we note $f$ by $f_{|f|}$ and note $H(f_j)$ by $H_j$.
\item Probe the backtracking function using the hash value $H_{|f|}=H(f)$,
retrieving a character $\beta_j$ (actually retrieving $f_{\Sigma}(\beta_j)$ then use the reverse
mapping $f^{-1}_{\Sigma}$ to get $\beta_j$). Then we do $|f|-1$ steps, 
computing for each $j\in[1,|f|-1]$ the hash value
$H_{j-1}=H_j-r^{f_{\Sigma}(\beta_j)}$ and probe the
backtracking function using the hash value $H_{j-1}$ retrieving the character
$\beta_{j-1}$. At the end of the $|f|-1$ steps we will have obtained a
sequence $s'_f=\beta_{|f|},\beta_{|f|-1},\ldots,\beta_1$ of characters. 
Suppose that $f\in {\cal F}$. When queried with the hash value $H_j$, 
the backtracking function would
return in this case the last character of the fingerprint
representation of $f$. Then $f_{j-1}=f_j/\{\beta_j\}$ would also represent
another fingerprint from ${\cal F}$. More generally we will have $f_j\in {\cal F}$
for every $j\in[1,|f|]$ with $f_j=\{\beta_1,\beta_2,\ldots,\beta_j\}$
\item The third step is to apply the method described in section~\ref{subsec:equal_test} in order to determine whether 
the set of characters in $s'_f$ equals the set of characters in $f$. If the two sets differ, we immediately conclude that
$f\notin {\cal F}$.
\item Finally we do a top-down traversal of the succinctly encoded trie described in section~\ref{subsec:succtrie} for the string $s'_f$. Here if the traversal fails before attaining a leaf, we immediately conclude that $f\notin {\cal F}$, otherwise conclude that $f\in {\cal F}$. 
\end{enumerate}
Now we can more precisely describe what is happening inside the data structure. We have to analyze two cases, the case $f\in {\cal F}$ and the case $f\notin {\cal F}$. For that we first prove the following lemmata:

\begin{lemma}
\label{lemma:query_step2}
Let $f\in {\cal F}$. Then 
\begin{enumerate} 
\item for each $j\in[1,|f|]$, $f_j\in {\cal F}$;
\item the string $s'_f$ is stored in the fingerprint trie. 
\end{enumerate}
\end{lemma}
\begin{proof}
The proof of fact 1 is by induction: $f$ is a valid fingerprint (by assumption) which means that the backtracking function returns the last character $\beta_{j}$ in the trie representation of $f$. Then we know that there exists some $f_{j-1}\in {\cal F}$ such that $f_{j-1}\cup \{\beta_j\}\in{\cal F}$. 
The base case of the induction is for $j=1$ (fingerprint consists of a single character $\beta_1$) in which case we clearly have a child of the fingerprint root labeled with character $\beta_1$.

The proof of fact 2 can also be obtained by induction. Assume that the assertion is true for a fingerprint $f_{j-1}$ of length $j-1$. 
Then it can be proved for a fingerprint $f_j$ of length~$j$, i.e. the assumption says that the sequence $s'_{f_{j-1}}=\beta_1,\beta_1,\ldots,\beta_{j-1}$ 
forms a permutation of $f_{j-1}$. We know that the backtracking function returns a character $\beta_j$ which is the last character 
of the representation of $f_j$ in the fingerprint trie and that there exists a fingerprint $f_{j-1}$ of size $j-1$ such that 
$f_{j-1}\cup \{\beta_j\}\in{\cal F}$. As we know that fact~2 is true for $f_{j-1}$, it means that the sequence 
$s'_{f_{j-1}}=\beta_1,\beta_2,\ldots,\beta_{j-1}$ of distinct symbols is a permutation of $f_{j-1}$.
Hence, by adding the character $\beta_j\notin f_{j-1}$ to the sequence we obtain a permutation of $f_j$.  
\end{proof}

From there we can get the following lemma: 
\begin{lemma}
If $f\in {\cal F}$ then the query successfully detects that $f\in {\cal F}$ and returns a positive answer.
\end{lemma}
\begin{proof}
By assumption $f\in {\cal F}$, which means by fact 2 of lemma~\ref{lemma:query_step2} that step 2 returns a sequence $s'_f$ which is a permutation of the set $f$. That means that step 3 will return a positive answer. It remains to be proven that step 4 is also successful. 
Moreover by fact 1 of lemma~\ref{lemma:query_step2}, step 4 will also be successful as step 4 traverses the fingerprint trie top-down where at each step it reaches a valid fingerprint $f_j$. 
\end{proof}
\begin{lemma}
Assuming that $f\notin {\cal F}$, either step~3 or step~4 will successfully detect that $f\notin {\cal F}$ and the query returns a negative answer.
\end{lemma}
\begin{proof}
The proof is by contradiction. Suppose that step 4 has concluded that $f\in{\cal F}$. Then steps 3 tells us that we have a sequence of $j$ characters $s'_f=\beta_0,\beta_1,\ldots,\beta_{|f|-1}$ which is a permutation of $f$ and that moreover by successfully traversing the trie in step 4 we deduce that $f\in {\cal F}$ which contradicts the premise that $f\notin {\cal F}$.
\end{proof}
Thus, we get the following theorem:
\begin{theorem}
\label{theo:succ_exist_DS}
The set of ${\cal F}$ of fingerprints of a sequence $s=s_1..s_n$ can be represented using a data structure that occupies $|{\cal F}|(2\log\sigma+\log_2 e)(1+o(1))$ bits. Given a set of characters $f$ the data structure is able to determine whether $f\in {\cal F}$ (existential queries) in time $O(|f|)$. 
\end{theorem}
We can also use the data structure to answer to report queries. However, in this case, because of the need to store pointers to occurrences, 
the representation will no longer be succinct (a pointer needs $\Omega(\log n)$ bits to be represented). We note that for each fingerprint, 
we can just store the list of maximal locations in the sequence using $2\log n$ bits for each element giving a total of $O(|{\cal L}|\log n)$ bits. 
However, a more space efficient approach is to use the suffix tree and for each fingerprint store a list of pointers to named copies in the suffix tree. 
This reduces the space to $O((n+|{\cal L}_C|)\log n)$ bits. Moreover, reporting the locations of the $\occ$ named copies from the suffix tree takes optimal $O(\occ)$ time as it consists in traversing a subtree with at most $\occ$ leaves and $\occ-1$ internal nodes.  
\begin{theorem}
\label{theo:fast_rep_DS}
Given a sequence $s=s_1..s_n$ of characters we can in time $O(n+|{\cal L}_c|\log\sigma)$ build a data structure that occupies $O((n+|{\cal L}_C|)\log n)$ bits of space such that given a fingerprint $f\in{\cal F}$ the data structure is able to report all the $\occ$ maximal locations in $s$ corresponding to $f$ in time $O(|f|+\occ)$. 

\end{theorem}

\section{Identifying Fingerprints in Less Space}
\label{less}

The result of theorem~\ref{theo:build_names} names all fingerprints of $s$ 
in time $\Theta(2\sigma+|{\cal L}_{C}|\log\sigma)$ while using
$O((|{\cal L}_{C}|+|{\cal F}|\log\sigma)\log n)$ bits of working space during the building.
The value $|{\cal L}_{C}|$ in the working space can dominate the value $|{\cal F}|\log\sigma$ 
when $|{\cal F}|\ll |{\cal L}_{C}|$.
When we need to build a data structure for report queries, then the value $|{\cal L}_{C}|$ is also
presented in the final size of required space and hence this presence in building space is unavoidable. 
However, when we only need to answer to existential queries, then the final data structure will use space 
of $O(|{\cal F}|\log\sigma)$ bits only. In this case it would be desirable to reduce the construction time as well. 
In this section, we show how to compute the set ${\cal F}$ in time $O(|{\cal L}|\log \sigma)$, 
but using space of $O(|{\cal F}|\log \sigma\log n)$ bits only. 

The original naming algorithm of~\cite{AmirALS03} is convenient for our 
purpose as it does the naming online without the need to carry the list of fingerprint 
changes (which is essentially equivalent to ${\cal L}$) until the end of the construction. 
The complexity of the algorithm of~\cite{AmirALS03} is $O(n\sigma\log n\log\sigma)$. 
The $\log n$ factor comes from the complexity of the use of binary search tree which is responsible
for the following task: given a pair of names $(\subname_0,\subname_1)$ at level $i$, find whether there is 
a unique name $\upname$ at level $i+1$ associated with the pair and 
if not add a new unique name $\upname$, associate it with the pair 
$(\subname_0,\subname_1)$ and add it to the binary search tree. 
This complexity of the naming algorithm was improved in~\cite{Ours2006,KR08b} 
from $O(n\sigma\log n\log\sigma)$ to just $O(|{\cal L}|\log \sigma)$ by the following way.
\begin{enumerate}
\item Notice that the naming has to deal only with $|{\cal L}|$ fingerprint changes instead of $n\sigma$.
 This reduces the factor $n\sigma$ to $|{\cal L}|$. 
\item Deferring the naming process until all the fingerprint changes have been recorded. Then using radix sort, the process 
time of giving unique names at level $i+1$ to pairs of names from level $i$ is
reduced to constant time per pair. This dispenses from the use of the binary search tree and reduces the factor $\log n$ to just~$1$. 
\end{enumerate}
This is the approach used in theorem~\ref{theo:build_names} and described in section~\ref{newnaming}. 

Our approach to improve~\cite{AmirALS03} is to notice that the binary search tree can be replaced 
with any hash table implementation which will change the time per operation from worst-case $O(\log n)$ 
to randomized expected $O(1)$. By this change the query time reduces to expected $O({\cal L}\log \sigma)$, but contrary to 
theorem~\ref{theo:build_names}, the building space remains as small as in~\cite{AmirALS03}, as we do not need
to record the fingerprint changes during the building process. More precisely during the naming process we need only 
to maintain at most $|{\cal F}|\log\sigma$ names (each fingerprint  might incur at most $\log\sigma$ names, one name at each level), 
which have been attributed so far. These names are recorded in a hash table which will use $O(|{\cal F}|\log\sigma\log n)$ bits of space. 
 
Thus, we have proven the following theorem:

\begin{theorem}
\label{theo:rand_build}
The set ${\cal F}$ of fingerprints of a sequence $s=s_1..s_n$ can be
computed in expected time $O(n+|{\cal L}|\log\sigma)$ time using $O((n+|{\cal F}|\log\sigma)\log n)$
bits of working space. 
\end{theorem}


\section{Randomized Identification Using a Monte Carlo Algorithm}
\label{random}

We now briefly sketch our construction algorithm that constructs the set of fingerprints ${\cal F}$ of 
the sequence $s$,  using only $O(|{\cal F}|\log n)$ bits ($O(|{\cal F}|)$ words) 
of temporary space and running in time $O(|{\cal L}|)$. While this approach might fail with an extremely small probability (the 
approach is said to be Monte Carlo or MC for short),  it might still be useful in case one wishes to get approximate statistics 
 on fingerprints: counting the total number of distinct fingerprints, or counting 
 the total number of strings having a given fingerprint, etc. 

To name the fingerprints we use use hash values of size $\Theta(\log n)$ bits. 
The hash values are computed using polynomial hash functions as described 
 in section~\ref{subsec:poly_hash}. 

Like in the previous section, the naming will be done online: we do not need not to 
store the fingerprint changes during the naming process. 
 Unlike the method described in the previous section, the fingerprint names will
 not be assigned deterministically, but will instead be assigned using
 hash values which could collide with extremely small probability.
 More specifically, in order to identify the existence of a fingerprint
we will use the polynomial hash functions as described in section ~\ref{subsec:poly_hash}
on the whole fingerprint. The polynomial
 hash function will be computed modulo $P$, where $P$ is a prime
 selected such that $P>n^cn^2\sigma^3)$. The chosen value of $P$ will ensure that each
 fingerprint will be mapped to a distinct value with probability at
 least $n^{-c}$. This can easily be seen: we have $|{\cal F}|<n\sigma$
 which implies that $|{\cal F}|^2<n^2\sigma^2$. Given that the polynomials
 are of degree at most $\sigma$, we can deduce that the probability of
 collision is at most
 $\frac{|{\cal F}|^2\sigma}{2P}<\frac{n^2\sigma^3}{n^cn^2\sigma^3}=n^{-c}$.
 \\ We now describe our algorithm in more detail. We assume that a set ${\cal S}$ of fingerprints can be
represented as a list $L = (\alpha_1,\alpha_2, \ldots \alpha_p)$ of
distinct characters such that $S= \{f_1,f_2,\ldots,f_p \} \mbox{ where
} f_i = \cup_{1 \leq j \leq i} \{ \alpha_j \}.$ 
We randomly choose a number $r\in [0,P]$ and the random hash function 
$H_r$ will be such that: $$H_r(f_i)=\sum_{1 \leq j \leq i}(r^{f_{\Sigma}(\alpha_j)})$$
The number $H_r(f_i)$ will be the unique name associated with the fingerprint $f_i$.
Now observe that $H_r(f_i)=H_r(f_{i-1})+r^{f_{\Sigma}(\alpha_i)}$. 
Thus computing the label
of $f_i$ can be done online using constant number of arithmetic operations based on $\alpha_i$ and $H_r(f_{i-1})$. 
In order to maintain the set of already processed fingerprints, we use a dynamic
hash table (for example using the MC real time dynamic hashing method
 described in~\cite{DH90}) that records the names of already processed fingerprints. Each time we generate the name 
of the fingerprint associated with a given maximal location we probe the dynamic hash table to see if that name 
already exists and if not add it to the hash table. If we also need to maintain the set of maximal locations
along with the set of fingerprints, we just associate a list of maximal locations to each fingerprint and store that 
list as  satellite data associated to the fingerprint name stored in the hash table. When the name of the fingerprint 
associated to a maximal location already exists in the hash table, this maximal location is added to the list of maximal 
locations associated with the fingerprint name in the hash table. If the fingerprint name did not already exist in the hash table, we add the name 
to hash table and associate a list of maximal locations which contains only the maximal location corresponding 
to the newly added fingerprint. 

In conclusion, we have proven the following theorem:

\begin{theorem}
\label{theo:MC}
The set ${\cal F}$ of fingerprints of a sequence $s=s_1..s_n$ can be
probabilistically computed in time $O(n+|{\cal L}|)$ using $O((n+|{\cal F}|)\log n)$
bits of working space. Moreover the set of maximal locations ${\cal
L}$ can be probabilistically determined in time $O(n+|{\cal L}|)$
using $O((n+|{\cal L}|)\log n)$ bits of working space.
The error rate probability can be made to $O(n^{-c})$ for any constant $c$. 
\end{theorem}


\nocite{AKOF99b}
\nocite{BCMR05}
\bibliographystyle{plain}

\end{document}